\documentclass{article}

\usepackage{amsthm}
\usepackage{graphics,graphicx}

\begin{document}
\title{Significant inference and confidence sets for graphical models }
\author{ Koldanov P.A.\footnote{pkoldanov@hse.ru}, Koldanov A.P.
\\
NRU HSE, Nizhny Novgorod, Russia  
%\\
%и университет Флориды, США
}
%\UDC{519.2}
%etc.
\date{}
\newtheorem{definition}{Definition}[section]
\newtheorem{teo}{Theorem}[section]
\newtheorem{lemma}{Lemma}[section]
\newtheorem{example}{Example}[section]
\newtheorem{corollary}{Corollary}[section]
\newtheorem{note}{Note}[section]
\maketitle
\abstract{}
%Рассматривается задача выделения статистически  значимых выводов о структуре графической модели и связанная с этим задача построения доверительного множества для графической модели. Доказано, что процедура построения доверительного  множества эквивалентна процедуре  одновременной проверки и гипотез и альтернатив о составе графической модели.
%Обсуждены некоторые варианты построения процедуры одновременной проверки и гипотез и альтернатив
%Показано, что при выполнении условия свободной комбинации гипотез и альтернатив простое обобщение метода замыкания приводит к одношаговым процедурам одновременной проверки и гипотез и альтернатив.
%Проанализирована структура доверительного множества для графической модели и показано, как доверительное множество приводит к разделению выводов о графической модели на статистически значимые и незначимые или зону неопределенности. 
%Результаты общего характера детализированы для задачи построения и анализа структуры доверительных множеств для undirected Gaussian graphical model selection.
%Приведены примеры разделения выводов о составе undirected Gaussian graphical model на значимые и зону неопределенности и проведено сравнение полученных результатов с известными результатами, полученными с помощью SINful approach to undirected Gaussian graphical model selection.

The problem of identifying statistically significant inferences about
the structure of the graphical model is considered, along with the related
task of constructing a confidence set for a graphical model. It has been
proven that the procedure for constructing such set is equivalent
to the procedure for simultaneous testing of hypotheses and alternatives
regarding the composition of the graphical model. Some variants of the
simultaneous testing of hypotheses and alternatives are discussed. It is
shown that under the condition of free combination of hypotheses and
alternatives, a simple generalization of the closure method leads to singlestep
procedures for simultaneous testing of hypotheses and alternatives.
The structure of the confidence set for the graphical model is analyzed,
demonstrating how the confidence set leads to a separation of inferences
about the graphical model into statistically significant and insignificant
categories, or into an area of uncertainty. General results are detailed
by analyzing confidence sets for undirected
Gaussian graphical model selection. Examples are provided that illustrate
the separation of inferences about the composition of undirected Gaussian
graphical models into significant results and areas of uncertainty, and a
comparison is made with known results obtained using the SINful approach
to undirected Gaussian graphical model selection.


\section{Introduction}\label{introduction}
%Задачам построения, анализа и применения графических моделей уделяется в последнее время большое внимание. Различные аспекты этой проблемы изучались в \cite{Jordan_2004}, \cite{Horvath_2011}, \cite{Wasserman2014}, \cite{Stifanelli_2011}, \cite{Kalyagin_2020}.  

%Среди всех графических моделей наиболее популярными являются Гауссовские графические модели. Различные аспекты этой проблемы изучались в \cite{Liang_2015}, \cite{Ren_2015}.
%Хорошо известно, что задачи селекции графических моделей уделяется большое внимание. Вместе с тем исследование проблемы контроля неопределенности и выделения статистически значимых выводов начато только в самое последнее время \cite{Drton2004}, \cite{Drton2008}, \cite{Kalyagin_2020}, \cite{Koldanov_2023}, \cite{Wang2021}.
It is well known that the problem of selecting graphical models receives significant attention. At the same time, the investigation of issues related to uncertainty control and the extraction of statistically significant conclusions has only recently been initiated \cite{Drton2004}, \cite{Drton2008}, \cite{Kalyagin_2020}, \cite{Koldanov_2023}, \cite{Wang2021}.

The problem of dividing inferences into statistically significant and indeterminate categories in the context of undirected Gaussian graphical model selection was first posed in \cite{Drton2004}, \cite{Drton2008}, where the SINful approach was proposed. It is well known that a Gaussian graphical model is understood as a graph \( G = (V, E) \) with \( V = \{1, \ldots, N\} \) and \( E = \{(i,j) : \rho^{i,j} \neq 0\} \), where \( \rho^{i,j} \) is the partial correlation coefficient. The SINful approach is based on the simultaneous testing of hypotheses \( h_{i,j} : \rho^{i,j} = 0 \), which allows constructing a subgraph whose set of edges, with a specified probability, is contained within the set of edges of the true Gaussian graphical model. The SINful approach involves partitioning the set of computed p-values into three parts: S, I, N, where S is the set of statistically significant conclusions regarding the falsehood of hypotheses about zero partial correlation coefficients; I is the set of indeterminate conclusions; and N is the set of large p-values corresponding to edges that should be excluded. However, only conclusions from the set S can be considered statistically significant, since the construction of N relies on a heuristic threshold setting, as noted by the authors of these works. Therefore, whether conclusions from the set N are statistically significant or indeterminate remains an open question. Using the results obtained in this work, it is possible to provide an answer to this question. Additionally, in \cite{Drton2004}, \cite{Drton2008}, interesting examples are presented demonstrating applications of the SINful approach to dividing inferences into S, I, N categories, which are analyzed in detail in this work.

In \cite{Koldanov_2023} a method for constructing upper and lower confidence bounds for the edges of a threshold similarity graph with a fixed threshold is proposed. It is shown that such bounds lead to a division of inferences about the structure of the threshold similarity graph into statistically significant and insignificant (indeterminate) categories. For this division, it is suggested to simultaneously test hypotheses about the presence of an edge and hypotheses about its absence. It is demonstrated that applying the procedure of simultaneous hypothesis testing and its alternatives results in the formation of three sets: a set of statistically significant edges, a set of statistically significant pairs of vertices not connected by an edge (non-adjacent vertices), and a set of statistically insignificant inferences, i.e., pairs of vertices for which current observations do not allow making a statistically significant conclusion about the presence or absence of an edge between them. It is shown that, with a given probability, the set of statistically significant edges is contained within the true edge set of the threshold similarity graph, and simultaneously, the true edge set is contained within the union of the set of statistically significant edges and the set of vertex pairs for which current observations do not permit a statistically significant conclusion about edge presence or absence. These sets are referred to in \cite{Koldanov_2023} as lower and upper confidence bounds for the true edge set of the threshold similarity graph. However, the approach proposed in \cite{Koldanov_2023} which involves simultaneous testing of one-sided hypotheses and their alternatives is limited according to the definition of the threshold similarity graph and does not align with many well-known graphical models, particularly Gaussian graphical models.

In \cite{Wang2021}, a procedure is proposed for constructing upper and lower graphs such that the true graphical model lies between them with a specified probability. Similar to \cite{Koldanov_2023}, the upper graph is understood as a graph whose set of edges includes the set of true edges, while the lower graph is a graph whose set of edges is contained within the true edge set with a given probability. The uncertainty in the model selection procedure (also as in \cite{Koldanov_2023}) is measured by the difference between the upper and lower graphs.
However, the connection between confidence sets and statistically significant inferences has not been explored. Additionally, the procedure for constructing the confidence set is based on bootstrap techniques. In the present work, the problem of constructing confidence sets for Gaussian graphical models and the related task of identifying statistically significant inferences are considered from classical perspectives, similar to \cite{Drton2004} and \cite{Drton2008}.

In this work, the approach proposed in \cite{Koldanov_2023} is extended to construct confidence sets for more general types of graphical models, including both the problem of threshold similarity graph selection and undirected Gaussian graphical model selection. Similar considerations apply to covariance graph models, which are defined in terms of marginal independence rather than conditional independence. This development is based on new proofs of certain results from \cite{Koldanov_2023}, which, in our opinion, better highlight the connection between the task of distinguishing statistically significant and insignificant conclusions about the structure of graphical models and the problem of simultaneous hypothesis testing and alternative hypotheses of a more general form.
Moreover, it is shown that the procedure for dividing conclusions about the truth or falsehood of hypotheses and alternatives tested simultaneously into statistically significant and insignificant ones is equivalent to constructing two index sets such that, with a given probability, the set of indices corresponding to true hypotheses includes the set of rejected alternatives and is simultaneously contained within the set of accepted hypotheses.
The problem of constructing procedures for simultaneous hypothesis testing and alternatives is considered from the perspective of classical multiple testing and simultaneous inference theory \cite{Lehmann2005}, \cite{Hochberg_1987}. It is demonstrated that, under the condition of free combination of hypotheses and alternatives, a simple extension of the closure method leads to single-step procedures.
Upper and lower confidence bounds are constructed, which define a confidence set for the graphical model. The lower confidence bound consists of pairs of vertices between which no edge is present in all graphs within the confidence set (the set of non-adjacent vertices), while the complement to the upper confidence bound consists of pairs of vertices between which an edge is present in all graphs within this set (the set of adjacent vertices). The difference (gap) between the upper and lower confidence bounds represents an area of uncertainty and characterizes the size of the confidence set. The lower confidence bound and its complement to the upper bound define sets corresponding to statistically significant conclusions about the structure of the graphical model.
General results are detailed for constructing confidence sets for undirected Gaussian graphical model selection. Examples are provided illustrating how conclusions about the structure of an undirected Gaussian graphical model can be separated into significant parts and an area of uncertainty, along with a comparison of these conclusions with known results \cite{Drton2004}, \cite{Drton2008} obtained using the SINful approach to undirected Gaussian graphical model selection.

The article is organized as follows:  
Section \ref{Basic_definitions_and_notations} presents the main definitions and problem formulation;  
Section \ref{conf_set_mult_test_sign_inference} demonstrates the connection between confidence sets, procedures for simultaneous testing of hypotheses and alternatives that control the Family-Wise Error Rate (FWER), and the division of conclusions into significant and insignificant;  
Section \ref{conf_set_construction} discusses possible approaches to constructing procedures for simultaneous testing of hypotheses and alternatives, leading to the formation of confidence sets for graphical models;  
Section \ref{gaussian_graph_model} addresses the problem of Gaussian graphical model selection;  
Section \ref{conclusion} provides the conclusion;  
In the appendix (Section \ref{Appendix}), all proofs are included.

\section{Basic definitions and problem statement}\label{Basic_definitions_and_notations} 

Let $(X_1,\ldots,X_N)'$ be a random vector with distribution $f(y;\theta),y\in R^N, \theta\in\Theta, \Theta\subseteq R^M$, $\theta=(\theta_1,\ldots,\theta_M)$. %, $\vartheta$ --- мешающие параметры. 
%Структуру зависимости между компонентами вектора $(X_1,\ldots,X_N)'$ удобно представить в виде графической модели или графа $G=(V,E), V=\{1,\ldots,N\}$ - множество вершин, $E$ - множество ребер. 
%Одной из наиболее известных графических моделей является Гауссовская графическая модель. При этом предполагается, что распределение вектора $(X_1,\ldots,X_N)'$ является многомерным нормальным, ненаправленное ребро между вершинами $i,j$ добавляется в граф $G$ тогда и только тогда, когда $X_i,X_j$ условно зависимы, т.е. частный коэффициент корреляции $\rho^{i,j}\neq 0$. В этом случае $\Theta=\{(\rho^{i,j})_{N\times N}\}$ - множество симметричных положительно определенных матриц. 

The dependence structure between the components of the vector \((X_1, \ldots, X_N)'\) can be conveniently represented using a graphical model or a graph \(G = (V, E)\), where \(V = \{1, \ldots, N\}\) is the set of vertices, and \(E\) is the set of edges.  
One of the most well-known graphical models is the Gaussian graphical model. In this case, it is assumed that the distribution of the vector \((X_1, \ldots, X_N)'\) is multivariate normal. An undirected edge between vertices \(i\) and \(j\) is added to the graph \(G\) if and only if \(X_i\) and \(X_j\) are conditionally dependent, i.e., their partial correlation coefficient \(\rho^{i,j} \neq 0\).  
In this context, \(\Theta = \{ (\rho^{i,j})_{N \times N} \}\) is the set of symmetric positive definite matrices.
%Задание Гауссовской графической модели эквивалентно заданию $\{(i,j):\rho^{i,j}=0\}$.
%Для сокращения обозначений вместо параметров $\rho^{i,j},i,j=1,\ldots,N$ мы будем использовать параметры $\theta_1,\ldots,\theta_M, M=\frac{N(N-1)}{2}$, т.е. 
%$\theta_{l}=\rho^{i,j}$ при $l=N(i-1)-\frac{1}{2}j(j-1)+i-j$. Аналогичным образом можно рассматривать и другие графические модели, такие как covariance graph
%models, which are defined in terms of marginal independence rather than conditional independence, threshold similarity graph and so on.
The specification of a Gaussian graphical model is equivalent to defining the set \(\{(i,j): \rho^{i,j} = 0\}\).  
To simplify notation, instead of the parameters \(\rho^{i,j}\) for \(i,j=1,\ldots,N\), we will use parameters \(\theta_1, \ldots, \theta_M\), where \(M = \frac{N(N-1)}{2}\). Specifically,  
\[
\theta_{l} = \rho^{i,j} \mbox{ when } l = N(i-1) - \frac{1}{2} j (j-1) + i - j.
\]  
Similarly, other types of graphical models can be considered, such as covariance graph models (defined in terms of marginal independence rather than conditional independence), threshold similarity graphs, and so on.

Let $X=(X_i(s)),i=1,\ldots,N;s=1,\ldots,n$ be a sample of size $n, n>N$.
The problem considered in this work is to construct a confidence set of graphs $S(X)$, i.e. a set $S(X)$, satisfying the condition  
\begin{equation}\label{conf_set_GGM}
P_G(S(X)\supset G)\geq P^{\star},\ \ \forall G\in {\cal G}
\end{equation}
where $G$ is the true graphical model, ${\cal G}$ is the set of all possible graphical models, and $0<P^{\star}<1$ is a specified value.

Any undirected graph can be represented by a symmetric adjacency matrix \(A_{N \times N} = (a_{i,j})\). Thus, the set of graphs included in \(S(X)\) can be described by a set of symmetric adjacency matrices. Any adjacency matrix can be specified by defining its elements.

Let us introduce the following notations. 

$J=\{i; i=1,\ldots,M\}$ the set of parameter indices $\theta_1,\ldots,\theta_M$,%is a set of indexes pairs, 
%$$M=|J|=\frac{N(N-1)}{2}$$
%где через $|A|$ обозначено число элементов множества $A$.
 %of random variables $X_1,X_2,\ldots,X_N$, 
%- ????????? ??? ???????? ????????? ???????

$J_t(\theta)$ - set of indices such that $\theta_i\in \Theta_i,i=1,\ldots,M$%is a set of edges of the true threshold similarity graph, %indexes pairs of strongly connected random variables for the threshold $\gamma_0$, set of conclusive edges of the correlation graph, 
%$J_{e}(\theta,\gamma_0)\subset J$, %?.?. $J_{e}(\theta,\gamma_0)=\{(i,j):s_{i,j}=1\}$, 

$J_f(\theta)$ - set of indices such that $\theta_i\in\Theta_i^c,i=1,\ldots,M$, %is a set of pairs of vertexes of the true threshold similarity graph without edges. 
%$J_n(\theta,\gamma_0)\subset J$,

where $\Theta_i\subset R^1$, $\Theta_i,\Theta_i^c$ is a partition of $\Theta$.
In particular, in the task of constructing a Gaussian graphical model, \(\Theta_i = \{0\}\).  
It is evident that the sets \(J_{t}(\theta)\) and \(J_{f}(\theta)\) form a partition of the set \(J\) for any \(\theta \in \Theta\).

  %Для определенности будем предполагать, что такие тесты имеют вид
%$$\varphi_i=\left\{\ 
%							\begin{array}{ll}
%							1,& t_i>c_{2i}\\
%							0,& t_i\leq c_{2i}
%							\end{array},
%						\right.
%	\psi_i=\left\{\ 
%							\begin{array}{ll}
%							1,& t_i<c_{1i}\\
%							0,& t_i\geq c_{1i}
%							\end{array}
%				\right.
%$$ 
%где $t_i=T_i(x),i=1,\ldots,N; t_i\in R^1$,  константы $c_{1i},c_{2i}$ определяются из 
%$$
%P_{\theta}(T_i(X)>c_{2i})=P_{\theta}(T_i(X)<c_{1i})=\alpha,\ \ \theta\in[\Theta_i]\cap[\Theta_i^c] - 
%$$
%общая граничная точка гипотезы и альтернативы.

%Пусть $x=(x_i(t)),i=1,\ldots,N;t=1,\ldots,n$ - повторная выборка размера $n$ из распределения вектора $X$. 

%Учитывая, что задание ненаправленного графа эквивалентно заданию множества его ребер, задачу построения доверительного множества $S(X)$ можно сформулировать как задачу построения двух множеств индексов $L(X)$ и $U(X)$, удовлетворяющих условию
Given that specifying an undirected graph is equivalent to defining its set of edges, the problem of constructing a confidence set \(S(X)\) can be formulated as the task of constructing two index sets \(L(X)\) and \(U(X)\), satisfying the condition
\begin{equation}\label{simultaneous_upper_low_bound}
P_{\theta}(L(X)\subseteq J_t(\theta)\subseteq U(X)\geq P^{\star},\ \ \forall\theta\in\Theta
\end{equation} 
%Очевидно, задачу построения доверительного множества $S(X)$ можно было бы сформулировать как эквивалентную задачу построения двух множеств индексов $L(X)$ и $U(X)$ для $J_f(\theta)$. В \cite{Koldanov_2023} такого типа множества были названы нижней и верхней доверительными границами множества ребер, в \cite{Wang2021} - small and large graphs. 

It is evident that the task of constructing a confidence set \(S(X)\) can be equivalently formulated as the task of constructing two index sets \(L(X)\) and \(U(X)\) for \(J_f(\theta)\).  
In \cite{Koldanov_2023}, such types of sets are referred to as lower and upper confidence bounds for the set of edges, while in \cite{Wang2021} they are called small and large graphs. 

%Применительно к задаче идентификации Гауссовской графической модели множество $L(X)$ представляет собой множество пар вершин $(i,j)$, между которыми во всех графах из доверительного множества $S(X)$ не проведено ребро, т.е. элемент $a_{i,j}=0$ во всех соответствующих матрицах смежности. Множество $\overline{U}(X)=J\setminus U(X)$ представляет собой множество пар вершин $(i,j)$, между которыми во всех графах из доверительного множества $S(X)$ проведено ребро, т.е. элемент $a_{i,j}=1$ во всех соответствующих матрицах смежности. Множество $U(X)\setminus L(X)$ представляет собой множество пар вершин $(i,j)$ со следующим свойством: граф с  ребром $(i,j)$, и граф без ребра $(i,j)$ принадлежат $S(X)$, т.е. $U(X)\setminus L(X)=\{(i,j):a_{i,j}=0\mbox{ или }a_{i,j}=1\}$. Таким образом, число графов, которые входят в $S(X)$, или число допустимых графических моделей, равно $2^{|U(X)\setminus L(X)|}$, где $|A|$-число элементов множества $A$.  

In this work, we consider the formulation given by (\ref{simultaneous_upper_low_bound}).  
Applying this to the problem of identifying a Gaussian graphical model, the set \(L(X)\) represents pairs of vertices \((i,j)\) for which no edge is present in all graphs within the confidence set \(S(X)\); that is, the element \(a_{i,j} = 0\) in all corresponding adjacency matrices.  
The set \(\overline{U}(X) = J \setminus U(X)\) consists of pairs \((i,j)\) for which an edge is present in all graphs within \(S(X)\); that is, \(a_{i,j} = 1\) in all corresponding adjacency matrices.  
The set \(U(X) \setminus L(X)\) contains pairs \((i,j)\) such that graphs with and without the edge \((i,j)\) both belong to \(S(X)\). In other words, these are pairs where both \(a_{i,j} = 0\) and \(a_{i,j} = 1\) are possible within the confidence set.  
Therefore, the total number of graphs (or admissible graphical models) within \(S(X)\) is  
\[
2^{|U(X) \setminus L(X)|},
\]
where \(|A|\) denotes the number of elements in the set \(A\). 
\section{Confidence set, multiple testing and significant inferences}\label{conf_set_mult_test_sign_inference}
%Рассмотрим отдельно задачи построения $U(x),L(x)$ для множества индексов $J_t(\theta)$ истинных гипотез.
%\subsection{Верхние множества}
%Пусть $U(x)=\{i:\varphi_i(x)=0\},\ \ i=1,\ldots,N$ --- множество индексов принятых гипотез $h_i,\ \ i=1,\ldots,N$.
%Для построения доверительного множества $S(X)$ мы предлагаем одновременно проверять и гипотезы $h_i:\theta_i\in \Theta_i,i=1,\ldots,M$ и альтернативы $k_i:\theta_i\in \Theta_i^c=\Theta\setminus\Theta_i, \Theta_i\subseteq R^1$.
 
To construct the confidence set \(S(X)\), we propose simultaneously testing the hypotheses \(h_i: \theta_i \in \Theta_i\), \(i=1,\ldots,M\), and the alternatives \(k_i: \theta_i \in \Theta_i^c = \Theta \setminus \Theta_i\), where \(\Theta_i \subseteq R^1\). 
%ля определения подходящих матриц смежности, удоывлетворяющих условию (\ref{conf_set_GGM}), предлагается использовать процедуры одновременной проверки многих гипотез. 
%Более точно, в настоящей работе задачу построения доверительного множества $S(x)$ предлагается рассматривать с позиций одновременной проверки гипотез $h_i:\theta_i\in \Theta_i,i=1,\ldots,M$ и альтернатив $k_i:\theta_i\in \Theta_i^c=\Theta\setminus\Theta_i, \Theta_i\subseteq R^1$. 

%Предположим, что тесты $\varphi_1,\varphi_2,\ldots,\varphi_M; \psi_1,\psi_2,\ldots,\psi_M$ проверки гипотез $h_i$ и проверки альтернатив $k_i$ произвольного уровня $0<\alpha<1$ известны. 

Suppose that the tests \(\varphi_1, \varphi_2, \ldots, \varphi_M\) and \(\psi_1, \psi_2, \ldots, \psi_M\), which test the hypotheses \(h_i\) and the alternatives \(k_i\), are known at an arbitrary significance level \(0 < \alpha < 1\).

\begin{lemma}\label{upper_set_lemma_new}
%Если процедура $(\varphi_1,\varphi_2,\ldots,\varphi_M)$ одновременной проверки гипотез $h_i,\ \ i=1,\ldots,M$ контролирует $FWER\leq\alpha$ в сильном смысле, то
If the procedure \((\varphi_1, \varphi_2, \ldots, \varphi_M)\), which performs simultaneous testing of hypotheses \(h_i,\ i=1,\ldots,M\), controls the family-wise error rate \(FWER \leq \alpha\) in the strong sense, then
\begin{equation}\label{upper_set_lemma_new_equation}
P_{\theta}(U(X)\supseteq J_t(\theta))\geq 1-\alpha,\ \ \forall\theta\in\Theta
\end{equation}
where 
$U(x)=\{i:\varphi_i(x)=0\},\ \ i=1,\ldots,M$ is the set of indices of accepted hypotheses $h_i,\ \ i=1,\ldots,M$.
\end{lemma}
%Доказательство леммы \ref{upper_set_lemma_new} приведено в приложении. 
The proof of Lemma \ref{upper_set_lemma_new} is provided in the appendix.
%\begin{note}
%Если $M=1$, то (\ref{upper_set_lemma_new_equation}) означает следующее:
%\begin{enumerate}
%\item если $J_t(\theta)=\emptyset$, т.е. проверяемая гипотеза неверна, то она может быть как принята ($U(x)=\{1\}$), так и отвергнута ($U(x)=\emptyset$), т.е. верятность, стоящая в левой части неравенства (\ref{upper_set_lemma_new_equation}), равна 1. 
%\item если $J_t(\theta)=\{1\}$, т.е. проверяемая гипотеза верна, то вероятность того, что она будет принята тестом уровня $\alpha$, не меньше $1-\alpha$. 
%\end{enumerate} ;;
%Таким образом, при проверке одной гипотезы условие (\ref{upper_set_lemma_new_equation}) не содержит новой информации и может рассматриваться, как короткая запись смысла теста $\varphi$ уровня $\alpha$. Такая запись, на наш взгляд, удобна для описания смысла процедуры одновременной проверки $M\geq 2$ гипотез, контролирующей $FWER\leq\alpha$ в сильном смысле.    
%\end{note}
%Множество $\overline{U}(x)=\{i:\varphi_i(x)=1\},\ \ i=1,\ldots,M$ представляет собой множество индексов отвергнутых гипотез $h_i,i=1,\ldots,M$, т.е. множество статистически значимых выводов о том, какие гипотезы ложны. 
%Применительно к задаче идентификации Гауссовской графической модели множество $\overline{U}(x)$ представляет собой множество статистически значимых выводов о наличии ребра.
The set \(\overline{U}(x) = \{i : \varphi_i(x) = 1\}, \; i=1,\ldots,M\), represents the set of rejected hypotheses \(h_i, i=1,\ldots,M\), i.e., the set of statistically significant conclusions about which hypotheses are false.  
In the context of identifying a Gaussian graphical model, the set \(\overline{U}(x)\) corresponds to the set of statistically significant inferences regarding the presence of edges.
%\subsection{Нижние множества}
%Пусть 

\begin{lemma}\label{low_set_lemma_new}
%Если процедура $(\psi_1,\psi_2,\ldots,\psi_M)$  одновременной проверки альтернативных гипотез $k_i,\ \ i=1,\ldots,M$ контролирует $FWER\leq\alpha$ в сильном смысле, то
If the procedure \((\psi_1, \psi_2, \ldots, \psi_M)\), which performs simultaneous testing of the alternative hypotheses \(k_i,\ i=1,\ldots,M\), controls the family-wise error rate \(FWER \leq \alpha\) in the strong sense, then
\begin{equation}\label{lower_set_lemma_new_equation}
P_{\theta}(L(X)\subseteq J_t(\theta))\geq 1-\alpha,\ \ \forall\theta\in\Theta
\end{equation}
where $L(x)=\{i:\psi_i(x)=1\},\ \ i=1,\ldots,M$ is the set of rejected alternatives $k_i,\ \ i=1,\ldots,M$.
\end{lemma}
The proof of Lemma \ref{low_set_lemma_new} is provided in the appendix.
%Множество $L(x)=\{i:\psi_i(x)=1\},\ \ i=1,\ldots,M$ представляет собой множество статистически значимых выводов о том, какие альтернативные гипотезы ложны, что равносильно статистически значимым выводам об истинности соответствующих гипотез.
%Применительно к задаче идентификации Гауссовской графической модели множество $L(x)$ представляет собой множество статистически значимых выводов об отсутствии ребра, т.е. об условной независимости соответствующих компонент случайного вектора $(X_1,\ldots,X_N)'$.
The set \(L(x) = \{i : \psi_i(x) = 1\}, \; i=1,\ldots,M\), represents the set of statistically significant conclusions about which alternative hypotheses are false, which is equivalent to statistically significant conclusions about the truth of the corresponding hypotheses.  
In the context of identifying a Gaussian graphical model, the set \(L(x)\) corresponds to the set of statistically significant inferences regarding the absence of edges, i.e., about the conditional independence of the corresponding components of the random vector \((X_1, \ldots, X_N)'\).

Lemma \ref{upper_set_lemma_new} and Lemma \ref{low_set_lemma_new} provide the possibility of separately constructing \(L(x)\) and \(U(x)\).  
For the simultaneous construction of both sets \(L(x)\) and \(U(x)\), consider the procedure  
\[
\delta = (\varphi_1, \psi_1, \varphi_2, \psi_2, \ldots, \varphi_M, \psi_M),
\]
where  

- \(\varphi_i, i=1,\ldots,M\)  tests for the hypotheses \(h_i\),  

- \(\psi_i, i=1,\ldots,M\)  tests for the alternative hypotheses \(k_i\).

Let's define the sets
$$L(x)=\{i:\psi_i=1\}$$
$$U(x)=\{i:\varphi_i=0\}$$
Additionally, we require that
\begin{equation}\label{comp_condition_new}
\{i:\varphi_i=1,\psi_i=1\}=\emptyset,\ \ \forall x\in R^{N\times n}, \forall i=1,\ldots,M,
\end{equation}
\begin{teo}\label{theorem_for_upper_low_bounds_construction}
If 
\begin{enumerate}
%\item $(\varphi_1,\varphi_2,\ldots,\varphi_N)$ --- процедура одновременной проверки гипотез $h_i,\ \ i=1,\ldots,N$, контролирующая $FWER\leq\alpha$ в сильном смысле; \item $(\psi_1,\psi_2,\ldots,\psi_N)$ --- процедура одновременной проверки альтернативных гипотез $k_i,\ \ i=1,\ldots,N$, контролирующая $FWER\leq\alpha$ в сильном смысле; 
\item procedure $\delta=(\varphi_1,\psi_1,\varphi_2,\psi_2,\ldots,\varphi_M,\psi_M)$ which performs simultaneous testing of the hypotheses $h_1,k_1,h_2,k_2,\ldots,h_M,k_M$ controls $FWER\leq\alpha$ in strong sense $\alpha=1-P^{\star}$,
\item condition (\ref{comp_condition_new}) is satisfied,
\end{enumerate}
then
\begin{equation}\label{simult_upper_low_bounds_new_theorem}
P_{\theta}(L(X)\subseteq J_t(\theta)\subseteq U(X))\geq P^{\star}, \ \ \forall\theta\in\Theta
\end{equation}
\end{teo}
%Доказательство теоремы \ref{theorem_for_upper_low_bounds_construction} приведено в приложении.
%Из теоремы \ref{theorem_for_upper_low_bounds_construction} следует, что множества $L(X)$ и $\overline{U}(X)$ представляют собой множества статистически значимых выводов о том, какие из одновременно проверяемых гипотез $h_1,k_1,h_2,k_2,\ldots,h_M,k_M$ ложны. Множество $U(x)\setminus L(x)$ представляет собой множество статистически незначимых выводов или множество допустимых выводов или зону неопределенности. Применительно к задаче идентификации Гауссовской графической модели множество $U(x)\setminus L(x)$ представляет собой множество пар вершин, между которыми можно как провести ребро, так и не провести ребро при заданном уровне значимости. Все соответствующие графы будут принадлежать доверительному множеству $S(x)$. 

The proof of Theorem \ref{theorem_for_upper_low_bounds_construction} is provided in the appendix.  
From Theorem \ref{theorem_for_upper_low_bounds_construction} follows that the sets \(L(X)\) and \(\overline{U}(X)\) represent statistically significant conclusions about which of the simultaneously tested hypotheses \(h_1, k_1, h_2, k_2, \ldots, h_M, k_M\) are false.  
The set \(U(x) \setminus L(x)\) represents a set of statistically insignificant conclusions, or an admissible set of conclusions, or a zone of uncertainty.  
In the context of identifying a Gaussian graphical model, the set \(U(x) \setminus L(x)\) corresponds to pairs of vertices between which an edge can either be present or absent at the given significance level. All graphs consistent with these conclusions will belong to the confidence set \(S(x)\).
%\begin{proof}
%Так как $L(x)=\{i:\psi_i(x)=1\},L_e(x)=\{i:\varphi_i=0,\psi_i=1\}$, то $P_{\theta}(L_e(X)\subseteq L(X))=1$. (Фактически, $P_{\theta}(L_e(X)=L(X))=1$, так как $L(x)\setminus L_e(x)=\{i:\varphi_i=1,\psi_i=1\}$ и условие (\ref{comp_condition_new}) выполняется.) 

%Так как $\{i:\varphi_i=0,\psi_i=0\}\cup\{i:\varphi_i=0,\psi_i=1\}=\{i:\varphi_i=0\}$, то $U_e(x)=U(x)$.

%Поэтому утверждение леммы следует из лемм \ref{upper_set_lemma_new}, \ref{low_set_lemma_new}.
%\end{proof}
%\begin{note}
%Если $M=1$ то (\ref{simult_upper_low_bounds_new_theorem}) означает следующее: 
%\begin{enumerate}
%\item если $J_t(\theta)=\emptyset$, т.е. верна альтернативная гипотеза, то вероятность того, что альтернативная гипотеза будет принята, не меньше $1-\alpha$. При этом проверяемая гипотеза может быть как принята, так и отвергнута.
%\item если $J_t(\theta)=\{1\}$, то вероятность того, что будет принята истинная гипотеза, не меньше $1-\alpha$. При этом альтернативная гипотеза может быть как принята, так и отвергнута.
%\end{enumerate}
%Таким образом, условие (\ref{simult_upper_low_bounds_new_theorem}) означает, что при $M=1$ в любом случае вероятность принятия правильного решения не меньше $P^{\star}=1-\alpha$. При этом допускается возможность одновременного принятия и проверяемой гипотезы и её альтернативы, и исключается (условие (\ref{comp_condition_new})) их одновременное отвержение. Подчеркнем, что условие (\ref{comp_condition_new}) вообще говоря, не является необходимым для выполнения (\ref{simult_upper_low_bounds_new_theorem}).
%\end{note}

\begin{teo}\label{equivalence_mult_proc_suff_set}
%Любая процедура построения множеств $U(x),L(x)$, $L(x)\subset U(x)$ удовлетворяющих условию (\ref{simult_upper_low_bounds_new_theorem}), порождает процедуру одновременной проверки и гипотез $h_1,\ldots,h_M$ и альтернативных гипотез $k_1,\ldots,k_M$, контролирующую $FWER\leq\alpha$ в сильном смысле, где $\alpha=1-P^{\star}$. 
Any procedure for constructing sets \(U(x)\) and \(L(x)\), with \(L(x) \subseteq U(x)\) and satisfying condition (\ref{simult_upper_low_bounds_new_theorem}), generates a simultaneous testing procedure for hypotheses \(h_1, \ldots, h_M\) and alternative hypotheses \(k_1, \ldots, k_M\), which controls the family-wise error rate \(FWER \leq \alpha\) in the strong sense, where \(\alpha = 1 - P^{\star}\).
\end{teo}
%Доказательство теоремы \ref{equivalence_mult_proc_suff_set} приведено в приложении.
%Теоремы \ref{theorem_for_upper_low_bounds_construction}, \ref{equivalence_mult_proc_suff_set} показывают, что между множеством процедур построения доверительных множеств $S(x)$ и множеством процедур одновременной проверки и гипотез $h_1,\ldots,h_M$ и альтернатив $k_1,\ldots,k_M$ существует взаимно-однозначное соответствие.
The proof of Theorem \ref{equivalence_mult_proc_suff_set} is provided in the appendix.  
Theorems \ref{theorem_for_upper_low_bounds_construction} and \ref{equivalence_mult_proc_suff_set} demonstrate that there is a one-to-one correspondence between the set of procedures for constructing confidence sets \(S(x)\) and the set of procedures for simultaneous testing of hypotheses \(h_1, \ldots, h_M\) and alternatives \(k_1, \ldots, k_M\).
 
%Таким образом, задача построения доверительного множества $S(X)$ для графа $G$ уровня $P^{\star}$ эквивалентна задаче построения процедуры одновременной проверки и гипотез $h_{i}$ и их альтернатив $k_{i}$, которая контролирует FWER в сильном смысле на уровне $1-P^{\star}$.
%Более того, множество пар $(i,j)$, таких что между вершинами $i$, $j$ проведено (не проведено) ребро во всех графах из доверительного множества $S(X)$, соответствует множеству отвергнутых гипотез (альтернатив), и следовательно, представляет собой множество статистически значимых выводов о составе графической модели. 
Thus, the problem of constructing a confidence set \(S(X)\) for the graph \(G\) at level \(P^{\star}\) is equivalent to the problem of constructing a procedure for simultaneous testing hypotheses \(h_i\) and their alternatives \(k_i\), which controls the family-wise error rate (FWER) in the strong sense at level \(1 - P^{\star}\). Moreover, the set of pairs \((i,j)\) such that an edge between vertices \(i\) and \(j\) is present (or absent) in all graphs within the confidence set \(S(X)\) corresponds to the set of rejected hypotheses (or accepted alternatives), and therefore represents statistically significant conclusions about the structure of the graphical model.

\section{Simultaneous testing of hypotheses and alternatives}\label{conf_set_construction}

%Для построения доверительных множеств $S(x)$ можно использовать по крайней мере два подхода.
%Первый подход заключается в отдельном построении множеств $L(X), U(X)$, опираясь на леммы \ref{upper_set_lemma_new}, \ref{low_set_lemma_new}, и объединении полученных результатов. 
%Второй подход основан на теореме \ref{theorem_for_upper_low_bounds_construction} и предполагает построение процедуры одновременной проверки и гипотез и альтернатив.
%В настоящей работе мы кратко обсудим первый подход и сосредоточимся на втором подходе, который в большей степени соответствует специфике рассматриваемого класса задач.

To construct confidence sets \(S(x)\), at least two approaches can be used.  
The first approach involves separately constructing the sets \(L(X)\) and \(U(X)\), relying on Lemmas \ref{upper_set_lemma_new} and \ref{low_set_lemma_new}, and then combining the obtained results.  
The second approach is based on Theorem \ref{theorem_for_upper_low_bounds_construction} and involves constructing a procedure for simultaneous testing of hypotheses and alternatives.  

In this work, we briefly discuss the first approach and focus on the second approach, which more closely aligns with the specifics of the class of problems under consideration.

\subsection{Separate construction of sets $L(X), U(X)$}

Let the lower and upper bounds \(L(x)\) and \(U(x)\) be constructed such that \(L(x) \subseteq U(x)\), and the following condition holds: $P_{\theta}(U(X)\supset J_t(\theta)\geq P_1^{\star}$, $P_{\theta}(L(X)\subset J_t(\theta)\geq P_2^{\star}$.
Then
$$P_{\theta}(L(X)\subset J_t(\theta)\subset U(X))=$$
$$=1-P_{\theta}(\{U(X)\not\supset J_t(\theta)\}\cup\{J_t(\theta)\not\subset L(X)\})\geq P_1^{\star}+P_2^{\star}-1$$
%Если $P_{\theta}(U(x)\supset J_e(\theta,\gamma_0))\geq P_1^{\star}$, $P_{\theta}(J_e(\theta,\gamma_0)\supset L(x))\geq P^{\star}_2$, то 
%$$P_{\theta}(U(x)\supset J_e(\theta,\gamma_0)\supset L(x)) \geq P_1^{\star}+P_2^{\star}-1$$
%Следовательно, множества $L(X), U(X)$ удовлетворяют (\ref{simult_upper_low_bounds_new_theorem}) при $P^{\star}=P_1^{\star}+P_2^{\star}-1$.
%В соответствии с теоремой \ref{equivalence_mult_proc_suff_set} процедура построения $L(X), U(X)$ порождает процедуру одновременной проверки $2M$ гипотез $h_1,\ldots,h_M$, $k_1,\ldots,k_M$, контролирующей $FWER$ на уровне $2-P_1^{\star}-P_2^{\star}=\alpha_1+\alpha_2$, где $\alpha_1=1-P_1^{\star}$ - ограничение на $FWER$ процедуры построения $L(x)$, $\alpha_2=1-P_2^{\star}$ - ограничение на $FWER$ процедуры построения $U(x)$. 
%Заметим, что так как $\alpha_1+\alpha_2=\alpha=1-P^{\star}$, то $\alpha_1<\alpha, \alpha_2<\alpha$.  
%В частности, можно положить $P_1^{\star}=P_2^{\star}=\frac{P^{\star}+1}{2}=1-\frac{\alpha}{2}$.% Ниже обсуждаются оба этих подхода.
Therefore, the sets \(L(X)\) and \(U(X)\) satisfy (\ref{simult_upper_low_bounds_new_theorem}) with \(P^{\star} = P_1^{\star} + P_2^{\star} - 1\).
According to Theorem \ref{equivalence_mult_proc_suff_set}, the procedure for constructing \(L(X)\) and \(U(X)\) generates a simultaneous testing procedure for \(2M\) hypotheses \(h_1, \ldots, h_M\), \(k_1, \ldots, k_M\), controlling the Family-Wise Error Rate (\(FWER\)) at the level \(2 - P_1^{\star} - P_2^{\star} = \alpha_1 + \alpha_2\), where \(\alpha_1 = 1 - P_1^{\star}\) is the bound on the \(FWER\) of the procedure for constructing \(L(x)\), and \(\alpha_2 = 1 - P_2^{\star}\) is the bound on the \(FWER\) of the procedure for constructing \(U(x)\).
Note that since \(\alpha_1 + \alpha_2 = \alpha = 1 - P^{\star}\), it follows that \(\alpha_1 < \alpha\) and \(\alpha_2 < \alpha\).
In particular, one can set \(P_1^{\star} = P_2^{\star} = \frac{P^{\star} + 1}{2} = 1 - \frac{\alpha}{2}\).
%Рассмотрим применение процедуры Холма как наиболее отвергающей в классе монотонных процедур \cite{Gordon_2008}. 
%Процедура одновременного построения $L(X), U(X)$ заключается в проверке гипотез $h_1,\ldots,h_M$ процедурой Холма с контролем $FWER\leq\alpha_1$, проверке альтернатив $k_1,\ldots,k_M$ процедурой Холма с контролем $FWER\leq\alpha_2$ при условии $\alpha_1+\alpha_2=\alpha$. Эту процедуру обозначим DH($\alpha_1,\alpha_2$)(Double Holm).

Let's consider the application of the Holm procedure as the most rejective procedure within the class of monotone procedures \cite{Gordon_2008}.  
The simultaneous construction of \(L(X)\) and \(U(X)\) involves testing hypotheses \(h_1, \ldots, h_M\) using Holm's procedure with a family-wise error rate (FWER) control at level \(\alpha_1\), and testing alternatives \(k_1, \ldots, k_M\) with Holm's procedure controlling FWER at level \(\alpha_2\), under the condition that \(\alpha_1 + \alpha_2 = \alpha\).  
We denote this combined procedure as DH(\(\alpha_1,\alpha_2\)) (Double Holm).

Order the p-values of the hypotheses $h_1,\ldots,h_M$. 
$$\hat{p}_{h_{(1)}}\leq\hat{p}_{h_{(2)}}\leq\ldots\leq \hat{p}_{h_{(M)}}$$
Note that $\hat{p}_{h_{i}}+\hat{p}_{k_{i}}=1, \forall i=1,\ldots,M$. 
%Поэтому $$\hat{p}_{k_{(1)}}\geq\hat{p}_{k_{(2)}}\geq\ldots\geq \hat{p}_{k_{(M)}}$$
Then
\begin{enumerate}
\item set $\overline{U}(x)$ constructed by DH($\alpha_1,\alpha_2$) % одновременной проверки гипотез $h_1,\ldots,h_M$ и альтернатив $k_1,\ldots,k_M$, 
has the form $\overline{U}(x)=\{(1),(2),\ldots,(k)\}$, $k<M$ if
$$\{\hat{p}_{h_{(1)}}\leq\frac{\alpha_1}{M}\}\cap\{\hat{p}_{h_{(2)}}\leq\frac{\alpha_1}{M-1}\}\cap\ldots\cap\{\hat{p}_{h_{(k)}}\leq\frac{\alpha_1}{M-k+1}\}\cap\{\hat{p}_{h_{(k+1)}}>\frac{\alpha_1}{M-k}\}$$ or $\overline{U}(x)=\{(1),(2),\ldots,(M)\}$ if
$$\{\hat{p}_{h_{(1)}}\leq\frac{\alpha_1}{M}\}\cap\{\hat{p}_{h_{(2)}}\leq\frac{\alpha_1}{M-1}\}\cap\ldots\cap\{\hat{p}_{h_{(M)}}\leq\alpha_1\}\}$$
\item set $L(x)$ constructed by DH($\alpha_1,\alpha_2$) %одновременной проверки гипотез $h_1,\ldots,h_M$ и альтернатив $k_1,\ldots,k_M$,
has the form $L(x)=\{(M),(M-1),\ldots,(M-l+1)\}$, $l<M$ if
$$\{\hat{p}_{h_{(M)}}>1-\frac{\alpha_2}{M}\}\cap\{\hat{p}_{h_{(M-1)}}>1-\frac{\alpha_2}{M-1}\}\cap\ldots\cap\{\hat{p}_{h_{(M-l+1)}}>1-\frac{\alpha_2}{M-l+1}\}\cap\{\hat{p}_{h_{(M-l)}}\leq 1-\frac{\alpha_2}{M-l}\}$$ or $L(x)=\{(M),(M-1),\ldots,(1)\}$ if
$$\{\hat{p}_{h_{(M)}}>1-\frac{\alpha_2}{M}\}\cap\{\hat{p}_{h_{(M-1)}}>1-\frac{\alpha_2}{M-1}\}\cap\ldots\cap\{\hat{p}_{h_{(1)}}>1-\alpha_2\}\}$$
\end{enumerate}
It is evident that $k+l\leq M$.

%Таким образом, для построения $L(X), U(X)$ можно использовать любые процедуры одновременной проверки гипотез и альтернатив, контролирующие FWER на уровнях $1-P_1^{\star}, 1-P_2^{\star}$ в сильном смысле.  Однако можно предложить    
%Поэтому наиболее простой способ построения совместных верхних и нижних доверительных границ уровня $P^{\star}$ для $J_e(\theta,\gamma_0)$ сводится к построению $\mbox{UCB}(x,\gamma_0,P_1^{\star})$ и $\mbox{LCB}(x,\gamma_0,P_2^{\star})$ отдельно при условиии $P^{\star}=P_1^{\star}+P_2^{\star}-1$. В частном случае $P_1^{\star}=P_2^{\star}=\frac{P^{\star}+1}{2}=1-\frac{\alpha}{2}$. Очевидно, что для построения $\mbox{UCB}(x,\gamma_0,P_1^{\star})$ и $\mbox{LCB}(x,\gamma_0,P_2^{\star})$ отдельно можно использовать любые процедуры одновременной проверки гипотез, контролирующие FWER в сильном смысле на уровнях $\frac{\alpha}{2}$. %In this case $J_{ue}(x,\gamma_0)=\mbox{USE}(\gamma_0,\frac{P^{\star}+1}{2}), J_{le}(x,\gamma_0)=\mbox{LSE}(\gamma_0,\frac{P^{\star}+1}{2})$.
 %Легко видеть, что $|\mbox{UCB}(x,\gamma_0,\frac{P^{\star}+1}{2})|\geq|\mbox{UCB}(x,\gamma_0,P^{\star})|$ и $|\mbox{LCB}(x,\gamma_0,\frac{P^{\star}+1}{2})|\leq|\mbox{LCB}(x,\gamma_0,P^{\star})|$. Это означает, что зазор между верхней и нижней доверительными границами, который может быть использован для измерения неопределенности, возрастает. %This means that the gap between upper and lower confidence bounds which could be used to characterize uncertainty is increased. 

\subsection{Modified Bonferroni procedure.}
%Очевидно, что для отдельного построения $L(X), U(X)$ можно использовать процедуру Бонферрони $B(\alpha_1,\alpha_2)$ при условии $\alpha_1+\alpha_2=\alpha$. Такая процедура заключается в сравнении $\hat{p}_{h_i}$ с $\frac{\alpha_1}{M}$ и $1-\frac{\alpha_2}{M}$. Однако учитывая специфику задачи одновременной проверки и гипотез и альтернатив, можно предложить процедуру $B(\alpha_1+\alpha_2,\alpha_1+\alpha_2)=B(\alpha,\alpha)$, которая контролирует $FWER\leq\alpha$ и отвергает не меньшее число гипотез и альтернатив, чем процедура $B(\alpha_1,\alpha_2)$.
%Такого типа процедура проверки односторонних гипотез $h_i:\theta_i\leq 0$ и альтернатив $k_i:\theta_i>0,\ \ i=1,\ldots,M$ была предложена в \cite{Bauer_1986} в сязи с контролем направленных ошибок. Такая процедура имеет вид 
It is evident that for the individual construction of \(L(X)\) and \(U(X)\), one can use the Bonferroni procedure \(B(\alpha_1,\alpha_2)\) under the condition \(\alpha_1 + \alpha_2 = \alpha\).  
This procedure involves comparing \(\hat{p}_{h_i}\) with \(\frac{\alpha_1}{M}\) and \(1 - \frac{\alpha_2}{M}\). However, considering the specifics of the simultaneous testing problem for hypotheses and alternatives, it is possible to propose a procedure \(B(\alpha_1 + \alpha_2, \alpha_1 + \alpha_2) = B(\alpha, \alpha)\), which controls \(FWER \leq \alpha\) and rejects at least as many hypotheses and alternatives as the procedure \(B(\alpha_1, \alpha_2)\).

Such a type of procedure for testing one-sided hypotheses \(h_i: \theta_i \leq 0\) against the alternatives \(k_i: \theta_i > 0\), for \(i=1,\ldots,M\), was proposed in \cite{Bauer_1986} in connection with controlling directional errors. The procedure has the form:
\begin{equation}\label{modified_bonf_Bauer}
\delta=(\varphi_1,\psi_1,\ldots,\varphi_M,\psi_M)
\end{equation}
%Тесты проверки гипотез $h_i$ и альтернатив $k_i$ имеют вид
$$\varphi_i=\left\{\ 
							\begin{array}{ll}
							1,& T_i>c_{2i}\\
							0,& T_i\leq c_{2i}
							\end{array},
						\right.
	\psi_i=\left\{\ 
							\begin{array}{ll}
							1,& T_i<c_{1i}\\
							0,& T_i\geq c_{1i}
							\end{array}
				\right.,
$$ 
where constants are defined from
$$
P_0(T_i>c_{2i})=\frac{\alpha}{M},
P_0(T_i<c_{1i})=\frac{\alpha}{M}.
$$
Note that when \(\frac{\alpha}{M} < \frac{1}{2}\) (i.e., for \(M \geq 2\)), the constants satisfy \(c_{1i} < c_{2i}\), and the condition (\ref{comp_condition_new}) is fulfilled.  
In \cite{Bauer_1986}, it was proven that the procedure (\ref{modified_bonf_Bauer}) controls the family-wise error rate (FWER) at level \(\alpha\). Therefore, from Theorem \ref{theorem_for_upper_low_bounds_construction}, it follows that the procedure (\ref{modified_bonf_Bauer}) leads to the construction of the sets \(U(x)\) and \(L(x)\) at level \(1 - P^{\star}\) when testing one-sided hypotheses.  

The procedure (\ref{modified_bonf_Bauer}) was used in \cite{Koldanov_2023} for constructing upper and lower bounds of the edge set in the problem of identifying a thresholded proximity graph.

For testing arbitrary hypotheses, a simple generalization of the procedure (\ref{modified_bonf_Bauer}) can be proposed, which leads to the construction of the sets \(\overline{U}(x) = J \setminus U(x)\) and \(L(x)\), having the following form:\begin{equation}\label{sets_mod_proc_Bonf}
\begin{array}{l}
\overline{U}(x)=\{i:\hat{p}_{h_{i}}\leq\frac{\alpha}{M}\}\\
L(x)=\{i:\hat{p}_{h_{i}}>1-\frac{\alpha}{M}\}.
\end{array}
\end{equation}
Let's show that under the condition (\ref{comp_condition_new}), the sets (\ref{sets_mod_proc_Bonf}) satisfy the condition:$$P_{\theta}(L\subset J_t\subset U)\geq \alpha=1-P^{\star}.$$
%модифицированная процедура Бонферрони проверки односторонних гипотез и альтернатив легко обобщается на случай проверки произвольных гипотез и альтернатив. 

%Пусть 
%$$J_t(\theta)=\{i_1,\ldots,i_k\}, J_f(\theta)=J\setminus J_t(\theta),$$
%$$L=\{i:\psi_i=1\}, U=\{i:\varphi_i=0\}.$$
%Тогда
%$$P_{\theta}(L\subset J_t\subset U)=1-P_{\theta}(L\not\subset J_t\mbox{ или }J_t\not\subset U)=$$
%$$=1-P_{\theta}(\exists j\in L\cap J_f\mbox{ или }\exists k\in J_t\cap\overline{U})\geq 1-P_{\theta}(\exists j\in L\cap J_f)-P_{\theta}(\exists k\in J_t\cap\overline{U})\geq$$
%$$\geq 1-\frac{k_1\alpha}{N}-\frac{k_2\alpha}{N}$$
%где 
%$k_1$--- число индексов, попавших в $L\cap J_f$,

%$k_2$--- число индексов, попавших в $J_t\cap\overline{U}$.

%Так как $k_1+k_2\leq N$, то 
%$$P_{\theta}(L\subset J_t\subset U)\geq 1-\alpha$$  

%\subsection{Модифицированная процедура Бонферрони. Общий случай}
Consider arbitrary hypotheses \(h_i: \theta_i \in \Theta_i\) and alternatives \(k_i: \theta_i \in \Theta_i^c\), for \(i=1,\ldots,M\).  
Let \(\varphi_i\) and \(\psi_i\) be tests for the hypotheses \(h_i\) and the alternatives \(k_i\), respectively, each at level \(\frac{\alpha}{M}\).

Suppose that:
$$J_t(\theta)=\{i_1,\ldots,i_k\}, J_f(\theta)=J\setminus J_t(\theta),$$
$$L=\{i:\psi_i=1\}, U=\{i:\varphi_i=0\}, 0\leq k\leq M.$$
Then
$$P_{\theta}(L\subset J_t\subset U)=1-P_{\theta}(L\not\subset J_t\mbox{ or }J_t\not\subset U)=$$
$$=1-P_{\theta}(\exists j\in L\cap J_f\mbox{ or }\exists k\in J_t\cap\overline{U})\geq 1-P_{\theta}(\exists j\in L\cap J_f)-P_{\theta}(\exists k\in J_t\cap\overline{U})\geq$$
$$\geq 1-\frac{k_1\alpha}{M}-\frac{k_2\alpha}{M}$$
where 
$k_1$ is the number of indices that fall into $L\cap J_f$,

$k_2$ is the number of indices that fall into $J_t\cap\overline{U}$.

If condition (\ref{comp_condition_new}) is fulfilled then $k_1+k_2\leq M$ and therefore 
$$P_{\theta}(L\subset J_t\subset U)\geq 1-\alpha$$  
Following \cite{Romano2015}, we will refer to the procedure (\ref{sets_mod_proc_Bonf}) as the modified Bonferroni procedure and denote it by \(mB(\alpha)\).
%Подчеркнем, что в \cite{Bauer} показано, что модифицированная процедура Бонферрони доминирует процедуру Холма одновременной проверки $2M$ односторонних гипотез.  Вместе с тем в \cite{Bauer} высказаны и критические замечания к этой процедуре.
%\begin{enumerate}
%\item	Отмечено, что возможность одновременного принятия одной из проверяемых гипотез и её альтернативы приводит к противоречию. На наш взгляд это не только не приводит к противоречию, но, напротив, в большей степени соответствует по крайней мере некоторым из задач, возникающих на практике и означает, что у нас нет возможности сделать значимый вывод об истинности или ложности проверяемых гипотез и поэтому оба вывода являются допустимыми. Заметим, что в \cite{Lehmann2005} в таких случаях рекомендуется продолжать наблюдения, если это, конечно, возможно.
%\item	Другой не всегда желательный аспект заключается в том, что исключена возможность принятия решения о том, что истинное значение параметра равно 0. Вместе с тем можно указать много практических задач, в которых такое решение не имеет самостоятельного значения. К числу таких задач относится, в частности, задача построения и анализа сетевых или графовых структур, в которых между двумя вершинами графа проводится ребро, если мера связи поведения этих вершин превышает некоторый порог. Кроме того, можно обратить внимание на точку зрения, отстаиваемую в \cite{Tukey2000}.       
%\end{enumerate}

\subsection{Generalized Closure Method}
One of the most general methods for constructing procedures for simultaneous testing of hypotheses \(h_1, \ldots, h_K\) is the closure method \cite{Marcus1976}. To account for the specifics of the problem involving simultaneous testing of hypotheses and alternatives, we will use the following simple generalization of the closure method:  
if any hypotheses \(H_Q = \bigcap_{i \in Q} h_i = \emptyset\) (where \(Q\) is some subset of the indices of the hypotheses being tested), then set \(\varphi_Q \equiv 1\).  
It is straightforward to verify that this generalization leads to a procedure for simultaneous hypothesis testing that controls the Family-Wise Error Rate (FWER) in a strong sense. For completeness, the corresponding result is provided in the appendix.

In the case of simultaneous testing of hypotheses and alternatives \(h_1, k_1, h_2, k_2, \ldots, h_M, k_M\), the closure method involves testing all intersection hypotheses of the form  
\[
H_{J_1,J_2} = \left(\bigcap_{i \in J_1} h_i\right) \cap \left(\bigcap_{i \in J_2} k_i\right), \quad \forall J_1, J_2 \subseteq \{1, 2, \ldots, M\}.
\]  
Since the intersection hypotheses \(H_{J_1,J_2} = \bigcap_{i \in J_1} h_i \cap \bigcap_{i \in J_2} k_i\), for which \(J_1 \cap J_2 \neq \emptyset\), are empty, we set:
$$\varphi_{J_1,J_2}\equiv 1, \forall J_1,J_2\subseteq\{1,2,\ldots,M\}:J_1\cap J_2\neq\emptyset.$$

%\begin{definition}\label{free_combination_condition_unrestricted_sence}
%Говорят, что проверяемые гипотезы $h_1,\ldots,h_M$ удовлетворяют условию свободной комбинации в сильном смысле, если 
%\begin{equation}
%\begin{array}{l}
%\forall P\subset\{1,2,\ldots,M\} \mbox{ гипотеза пересечения } H_P=\bigcap_{i\in P}h_i\neq\emptyset\\
%\forall P_1,P_2\subset\{1,2,\ldots,M\}\mbox{ таких, что }P_1\neq P_2 \ \ \bigcap_{i\in P_1}h_i\neq \bigcap_{i\in P_2}h_i\\
%\forall P_1,P_2\subset\{1,2,\ldots,N\} \mbox{ таких что }P_1\neq P_2\mbox{ выполняется }H_{P_1}\neq H_{P_2}
%\end{array}
%\end{equation}
%\end{definition}
%\begin{definition}\label{free_combination_condition}
%Говорят, что проверяемые гипотезы $h_1,\ldots,h_M$ удовлетворяют  условию свободной комбинации в слабом смысле если 
%\begin{equation}
%  \begin{array}{l}
%		\forall P\subset\{1,2,\ldots,M\} \mbox{ гипотеза пересечения } H_P=\bigcap_{i\in P}h_i\neq\emptyset\\ 
%	\end{array}
%\end{equation}
%\end{definition}
Let hypotheses and alternatives \(h_1, k_1, h_2, k_2, \ldots, h_M, k_M\) satisfy the free combination condition \cite{Liu1996}, according to which 
%\begin{definition}\label{free_combination_condition_strong_sence}
%Говорят, что гипотезы $h_i$ против альтернатив $k_i,i=1,\ldots,M$ удовлетворяют условию свободной комбинации гипотез и альтернатив если 
\begin{equation}
\begin{array}{l}
\forall P\subseteq\{1,2,\ldots,M\} \mbox{ one has } \\
\bigcap_{i\in P}h_i\cap\bigcap_{i\in\{1,2,\ldots,M\}\setminus P}k_i\neq\emptyset\\
%2)&\forall P_1,P_2\subset\{1,2,\ldots,N\} \mbox{ таких что }P_1\neq P_2\mbox{ выполняется }H_{P_1}\neq H_{P_2}
\end{array}
\end{equation}
%\end{definition}

In this case, among all non-empty intersection hypotheses, the minimal ones \cite{Finner_2002} are the intersection hypotheses of the form  
\[
H_{J_1} = H_{J_1, \{1,\ldots,M\} \setminus J_1} = \left(\bigcap_{i \in J_1} h_i\right) \cap \left(\bigcap_{i \in \{1,\ldots,M\} \setminus J_1} k_i\right).
\]
Each of these hypotheses corresponds to an intersection of \(M\) parametric regions. The number of such intersections and consequently, the number of hypotheses is \(2^M\). Under the free combination condition for hypotheses and alternatives, all hypotheses \(H_{J_1}\) are non-empty.
Any hypothesis \(H_{J_1,J_2} = \bigcap_{i \in J_1} h_i \cap \bigcap_{i \in J_2} k_i\), such that \(J_1 \cup J_2 \subseteq \{1, 2, \ldots, M\}\) and \(J_1 \cap J_2 = \emptyset\), can be represented as the union of the corresponding hypotheses \(H_{J_i}\), namely:
$$H_{J_1,J_2}=\bigcap_{i\in J_1}h_i\cap\bigcap_{i\in J_2}k_i=\bigcup_{\pi\subseteq\{i_1,\ldots,i_k\}}H_{J_1\cup\pi}=\bigcup_{\pi\subseteq\{i_1,\ldots,i_k\}}\left(\bigcap_{i\in J_1\cup\pi}h_i\cap\bigcap_{i\in \{1,\ldots,M\}\setminus \{J_1\cup\pi\}}k_i\right)$$
where $\{i_1,\ldots,i_k\}=\{1,2,\ldots,M\}\setminus(J_1\cup J_2)$. 
Thus, under the free combination condition for hypotheses and alternatives, the hypotheses \(H_{J_1,J_2}\) are non-empty for all \(J_1, J_2 \subseteq \{1,\ldots,M\}\) with \(J_1 \cap J_2 = \emptyset\).  
To test the intersection hypotheses \(H_{J_1,J_2}\), we will use union-intersection tests, in which each hypothesis or alternative belonging to \(J_1 \cup J_2\) is tested at the same significance level.    
\begin{teo}\label{closure_equivalent_bonferroni_theorem}
If
\begin{enumerate}
\item the hypotheses \(h_i\) and the alternatives \(k_i\) satisfy the free combination condition,% (см. определение \ref{free_combination_condition_strong_sence}), 
\item to test the intersection hypothesis \(H_{J_1,J_2}\), a union-intersection method is used,
\item the condition (\ref{comp_condition_new}) $\forall x\in R^{N\times M},i=1,\ldots,M\ \ \{i:\varphi_i=1,\psi_i=1\}=\emptyset$ is satisfied, 
\end{enumerate}
then, the procedure for simultaneous testing of hypotheses \(h_i\) and alternatives \(k_i\), constructed using the generalized closure method with control of the Family-Wise Error Rate (FWER) at level \(\alpha\), is a single-step procedure and has the form:
$$\left(\varphi_1,\ldots,\varphi_M,\psi_1(x),\ldots,\psi_M(x)\right)$$
where 
$$\varphi_i(x)=\left\{\
									\begin{array}{ll}
									1,& \hat{p}_{h_i}<\alpha(M)\\
									0,& \hat{p}_{h_i}\geq\alpha(M)
									\end{array}
							 \right.,
$$
$$\psi_i(x)=\left\{\
									\begin{array}{ll}
									1,& \hat{p}_{k_i}<\alpha(M)\\
									0,& \hat{p}_{k_i}\geq\alpha(M)
									\end{array}
							 \right.
$$

$\alpha(M)$ is defined from equation
$$P_{\theta}(\min_{i\in\{1,\ldots,M\}}(\hat{p}_{h_i})<\alpha(M))=\alpha.$$ 
\end{teo}
Proof of the theorem \ref{closure_equivalent_bonferroni_theorem} is given in Appendix.% \ref{proof_closure_equivalent_bonferroni_theorem}.
\begin{corollary}
If p-values $\hat{p}_{h_i}$ are independent $\forall i=1,\ldots,M$ then $\alpha(M)=1-\sqrt[M]{1-\alpha}$.
\end{corollary}
\begin{corollary}
If the \(mB(\alpha)\) procedure is used to test the intersection hypotheses, then$\alpha(M)=\frac{\alpha}{M}$.
\end{corollary}

\section{Undirected Gaussian graphical model selection}\label{gaussian_graph_model}

It is well known that level $\alpha$ test for testing hypothesis $h_{i,j}:\rho^{i,j}=0$ against alternative $k_{i,j}:\rho^{i,j}\neq 0$ has the form \cite{Anderson_2003}:
\begin{equation}\label{Partial_correlation_test}
\varphi_{i,j}=\left\{\ 
\begin{array}{ll} 
0,&|r^{i,j}|\leq c_{i,j}\\
1,&|r^{i,j}|> c_{i,j}
\end{array}\right.
\end{equation} 
where 
$c_{i,j}$ is $(1-\alpha/2)$-quantile of the distribution with the following density function
\begin{equation}\label{density_function_sample_partial_correlation_test}
f(x)=\displaystyle \frac{1}{\sqrt{\pi}}\frac{\Gamma(n-N+1)/2)}{\Gamma((n-N)/2)}(1-x^2)^{(n-N-2)/2}, \ \ \ -1 \leq x \leq 1
\end{equation}
In \cite{Koldanov_2017} it was shown that test (\ref{Partial_correlation_test}) is UMPU.

It is evident that level $\alpha$ test for testing alternative $k_{i,j}:\rho^{i,j}\neq 0$ against hypothesis $h_{i,j}:\rho^{i,j}=0$ 
%можно записать в виде $\psi_{i,j}=1-\varphi_{i,j}$. При этом если тест $\varphi_{i,j}$ имеет уровень $\alpha$, то тест $\psi_{i,j}$ имеет уровень $1-\alpha$. 
%Поэтому тест уровня $\alpha$ проверки альтернативы  $k_{i,j}:\rho^{i,j}\neq 0$ против гипотезы $h_{i,j}:\rho^{i,j}=0$ 
has the form
\begin{equation}\label{Partial_correlation_test_alternative}
\psi_{i,j}=\left\{\ 
\begin{array}{ll} 
1,&|r^{i,j}|\leq c'_{i,j}\\
0,&|r^{i,j}|> c'_{i,j}
\end{array}\right.
\end{equation} 
where 
$c'_{i,j}$ is $\alpha/2$-quantile of the distribution with the density function (\ref{density_function_sample_partial_correlation_test}).
The widely accepted practice for applying the test (\ref{Partial_correlation_test}) is to use Fisher's z-transformation \cite{Anderson_2003}, \cite{Drton2004}.

Hypotheses \(h_{i,j}\), \(i,j=1,\ldots,N\), and alternatives \(k_{i,j}\), \(i,j=1,\ldots,N\), satisfy the free combination condition.  
For \(\alpha < \frac{1}{2}\), it holds that \(c'_{i,j} < c_{i,j}\), i.e., \(\{(i,j): \varphi_{i,j} = 1, \psi_{i,j} = 1\}\) is empty. Thus, the condition (\ref{comp_condition_new}) is satisfied.  
Let us consider testing the intersection hypothesis  
\[
H_{J_1,J_2} = \bigcap_{(i,j) \in J_1} h_{i,j} \cap \bigcap_{(l,s) \in J_2} k_{l,s}
\]
using union-intersection tests. Then, from Theorem \ref{closure_equivalent_bonferroni_theorem}, it follows that the procedure for simultaneous testing of hypotheses \(h_{i,j}\) and alternatives \(k_{i,j}\), constructed via the generalized closure method, is a single-step procedure and has the form:
\begin{equation}\label{modified_Bonfer_GGM}
(\varphi_{1,2},\psi_{1,2},\varphi_{1,3},\psi_{1,3},\ldots,\varphi_{N-1,N},\psi_{N-1,N})
\end{equation} 
where tests $\varphi_{i,j},\psi_{i,j}$ are level $\alpha(M)$ tests (\ref{Partial_correlation_test}),(\ref{Partial_correlation_test_alternative}) respectively. 
In terms of p-values, the procedure (\ref{modified_Bonfer_GGM}) has the form:
\begin{equation}\label{modified_Bonfer_GGM_through_indicators}
\left(I\left(\hat{p}^{\varphi}_{1,2}\leq 1-\sqrt[M]{1-\alpha}\right),I\left(\hat{p}^{\varphi}_{1,2}\geq \sqrt[M]{1-\alpha}\right),\ldots,I\left(\hat{p}^{\varphi}_{N-1,N}\leq 1-\sqrt[M]{1-\alpha}\right),I\left(\hat{p}^{\varphi}_{N-1,N}\geq \sqrt[M]{1-\alpha}\right)\right)
\end{equation} 
In accordance with Theorem \ref{theorem_for_upper_low_bounds_construction}, the procedure (\ref{modified_Bonfer_GGM_through_indicators}) results in the construction of upper and lower confidence bounds for the set \(\{(i,j): \rho^{i,j} = 0\}\). These bounds take the form:
$$U(x)=\{(i,j):\hat{p}^{\varphi}_{i,j}>1-\sqrt[M]{1-\alpha}\}; $$
$$L(x)=\{(i,j):\hat{p}^{\varphi}_{i,j}\geq \sqrt[M]{1-\alpha}\}$$
%Заметим, что результаты применения процедуры (\ref{modified_Bonfer_GGM_through_indicators}) несильно отличаются от результатов применения $mB(\alpha)$, в которой $\alpha(M)=\frac{\alpha}{M}$. 
The bounds \(U(x)\) and \(L(x)\) generate a confidence set \(S(x)\) for the undirected Gaussian graphical model, which serves as our model selection method.

Such a model selection method for the undirected Gaussian graphical model can be conveniently represented as a set of symmetric adjacency matrices \(A_{N \times N} = (A_{i,j})\), where \(A_{i,j} = A_{j,i}\), with the following property:
$$
A_{i,j}=\left\{\
				\begin{array}{ll}
				 1,& \hat{p}^{\varphi}_{i,j}\leq \alpha(M)\\
				 0,& \hat{p}^{\varphi}_{i,j}\geq 1-\alpha(M)\\
				 \epsilon_{i,j},& \alpha(M)<\hat{p}^{\varphi}_{i,j}<1-\alpha(M)
				\end{array}
			 \right.
$$
$\epsilon_{i,j}\in\{0,1\}$. The set \(\{(i,j):A_{i,j} = 1\} = \overline{U}(x)\) represents the set of significant inferences about the presence of edges.  
The set \(\{(i,j):A_{i,j} = 0\} = L(x)\) represents the set of significant inferences about the absence of edges.  
The set \(\{(i,j):A_{i,j} = \epsilon_{i,j}\} = U(x) \setminus L(x)\) corresponds to the set of insignificant inferences or the zone of uncertainty.  

By assigning values 0 or 1 to different elements \(\epsilon_{i,j}\) for pairs \((i,j) \in U(x) \setminus L(x)\), we obtain a confidence set \(S(x)\) for Gaussian graphical models.  
Any graph from this confidence set \(S(x)\) can be chosen as the true undirected Gaussian graphical model.  

The size of this set of graphs is \(2^{|U(x) \setminus L(x)|}\).
Another model selection method for the undirected Gaussian graphical model can be based on the application of the procedure \(DH(\alpha_1, \alpha_2)\), where \(\alpha_1 + \alpha_2 = \alpha\).
%Из условия (\ref{simult_upper_low_bounds_new_theorem}) теоремы \ref{theorem_for_upper_low_bounds_construction} следует, что это множество представляет собой доверительное множество для истинной графической модели уровня $P^{\star}=1-\alpha$.    
%\end{note}
\subsection{Examples}
We now apply our model selection methods to three well-known examples and compare our selected models to the results of SINfull approach \cite{Drton2004}.

In the analysis of examples, both \(\alpha(M) = \frac{\alpha}{M}\) and \(\alpha(M) = 1 - \sqrt[M]{1 - \alpha}\) were used. Since the obtained results do not differ, for simplicity of notation, below are the numerical values corresponding to the \(mB(\alpha)\) procedure.
\begin{example}
Let us consider example cork boring from \cite{Drton2004} where $N=4$, i.e. $M=6$.
Corresponding p-values $\hat{p}^{\varphi}_{i,j}$% тестов $\varphi_{i,j}$ проверки гипотез $h_{i,j}:\rho^{i,j}=0$ 
are equal to:
$$\hat{p}^{\varphi}_{1,2}=0.01;\hat{p}^{\varphi}_{1,3}=0.71;\hat{p}^{\varphi}_{1,4}=0.44;\hat{p}^{\varphi}_{2,3}=0.9;\hat{p}^{\varphi}_{2,4}=0.95;\hat{p}^{\varphi}_{3,4}=0.001. 
$$
\begin{enumerate}
\item $P^{\star}=0.99$ Then $\alpha=0.01$ $\frac{\alpha}{M}=0.0017$, $1-\frac{\alpha}{M}=0.9983$. Therefore, in accordance with the \(mB(\alpha)\) procedure, we obtain:  
$$L(x)=\emptyset$$
$$\overline{U}(x)=\{(3,4)\}$$
$$U(x)\setminus L(x)=\{(1,2);(1,3);(1,4);(2,3);(2,4)\}$$
The confidence set \(S(x)\) is schematically illustrated in Fig.\ref{Cork_boring_099}. Solid lines indicate edges that are present in all graphical models within the set $S(x)$. 
\begin{figure}[h!]
\centering
\includegraphics*[width=0.25\textwidth]{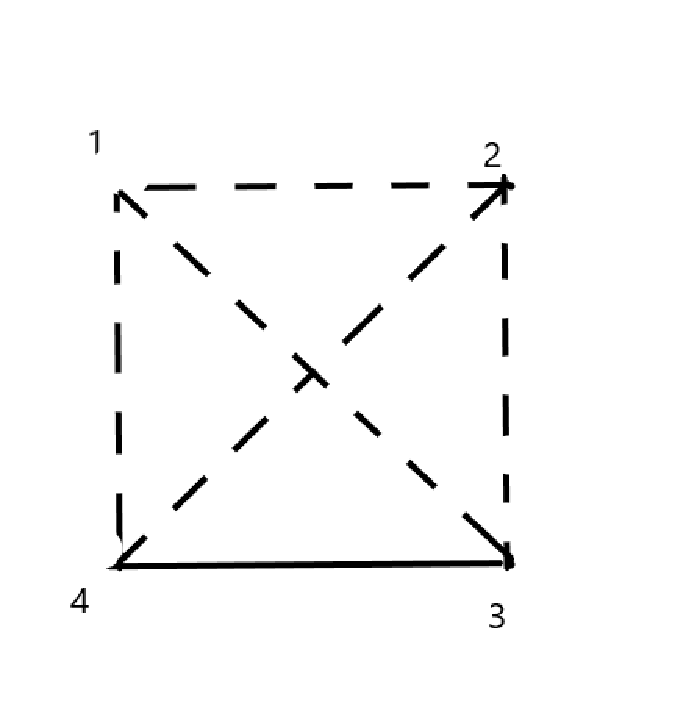}
\caption{Cork boring from \cite{Drton2004}, $N=4$. Solid lines correspond to significant edges, while dashed lines correspond to insignificant edges. $P^{\star}=0.99$
}\label{Cork_boring_099}
\end{figure}
The set of graphical models that can be obtained by replacing the dashed lines in the graph of Fig.\ref{Cork_boring_099} with either solid lines or no lines at all constitutes a confidence set for the true graphical model at the level \(P^{\star} = 1 - \alpha = 0.99\).  
The size of this confidence set is \(2^5 = 32\).

Applying the procedure \(DH\left(\frac{\alpha}{2}, \frac{\alpha}{2}\right)\) results in \(\overline{U}(x) = L(x) = \emptyset\), and  
\[
U(x) \setminus L(x) = \{(1,2); (1,3); (1,4); (2,3); (2,4); (3,4)\}.
\]  
The size of the confidence set \(S(x)\), constructed using the \(DH\left(\frac{\alpha}{2}, \frac{\alpha}{2}\right)\) procedure in this example, is \(2^6 = 64\).

\item $P^{\star}=0.9$. Then $\alpha=0.1$ $\frac{\alpha}{M}=0.017$, $1-\frac{\alpha}{M}=0.983$. Therefore, in accordance with the \(mB(\alpha)\) procedure, we obtain:    
$$L(x)=\emptyset$$
$$\overline{U}(x)=\{(1,2);(3,4)\}$$
$$U(x)\setminus L(x)=\{(1,3);(1,4);(2,3);(2,4)\}$$
The confidence set \(S(x)\) is schematically illustrated in Fig.\ref{Cork_boring}.% Сплошными линиями показаны ребра, которые присутствуют во всех графических моделях из $S(x)$. 
\begin{figure}[h!]
\centering
\includegraphics*[width=0.25\textwidth]{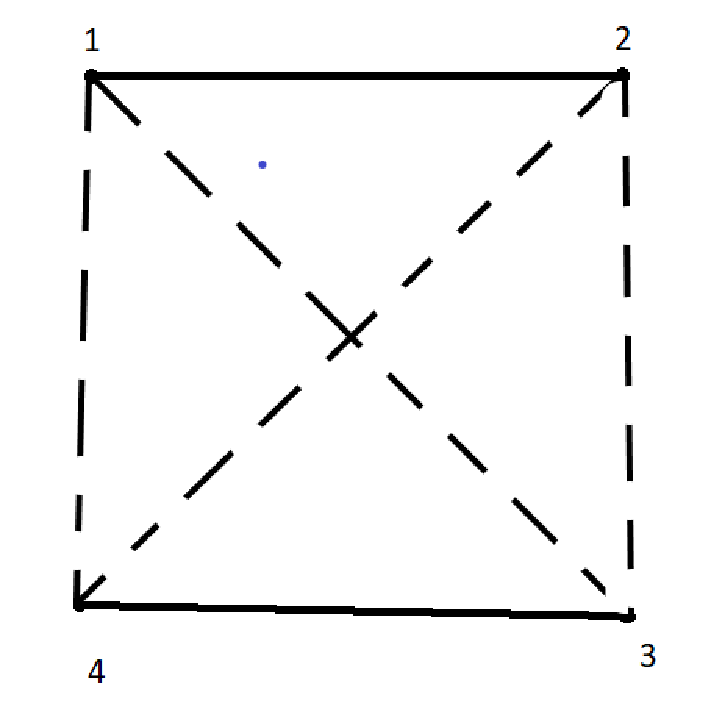}
\caption{Cork boring from \cite{Drton2004}, $N=4$. Solid lines correspond to significant edges, while dashed lines correspond to insignificant edges. $P^{\star}=0.9$
}\label{Cork_boring}
\end{figure}
The set of graphical models that can be obtained by replacing the dashed lines in the graph of Fig. \ref{Cork_boring} with either solid lines or no lines at all constitutes a confidence set for the true graphical model at the level $P^{\star}=1-\alpha=0.9$. The size of this confidence set is $2^4=16$.

Applying the procedure \(DH\left(\frac{\alpha}{2}, \frac{\alpha}{2}\right)\) leads to the same result.
\end{enumerate}

In \cite{Drton2004}, the p-values for hypotheses \(h_{i,j}:\rho^{i,j}=0\) are partitioned into groups \(S = (0.001; 0.01)\), \(I = (0.44)\), and \(N = (0.71; 0.9; 0.95)\). In our notation, the group \(S\) corresponds to the set of pairs \(\overline{U}\), group \(I\) corresponds to the set of pairs \(U \setminus L\), and group \(N\) corresponds to the set of pairs \(L\).  

This partition implies that hypotheses \(h_{i,j}:\rho^{i,j}=0\) with p-values 0.001 and 0.01, as well as the alternatives \(k_{i,j}:\rho^{i,j} \neq 0\) with p-values 0.29, 0.1, and 0.05, are rejected.  

Rejecting the alternative \(k_{i,j}:\rho^{i,j} \neq 0\) with p-value 0.29 using the \(mB(\alpha)\) procedure indicates that  
\[
0.29 = \frac{\alpha}{6} \rightarrow \alpha = 1.74,
\]
which exceeds any typical significance level \(\alpha \leq 1\). Therefore, such a rejection is not statistically significant at any standard significance level.  

Thus, this partitioning leads to \(P^{\star} = 0\). 
%, т.е. $P^{\star}_1=1-\alpha=0.4$ и $P(L(X)\subset J_t(\theta)\subset U(X))\geq 0.4$, где $L(x)=\{(2,3),(2,4)\},U(x)=\{(2,3),(2,4),(1,3),(1,4)\}$.    
%Таким образом, p-значения алтерниативных гипотез $k_{i,j}:\rho^{i,j}\neq 0$ равные $0.29;0,1; 0,05$ будут включены в множество $L(x)$, т.е. эти альтернативные гипотезы будут отвергнуты. Отвержение гипотезы $k_{i,j}:\rho^{i,j}\neq 0$ с p-значением $0.29$ процедурой Бонферрони означает, что $0.29<\frac{\alpha}{6}$, поэтому $\alpha>0.29*6=1,74$, т.е. $P^{\star}=1-\alpha=-0.74$ и $P(L(X)\subset J_t(\theta))\geq 0$, т.е. полученное утверждение не обладает какими-либо вероятностными свойствами. 

%Рассмотрим разбиение p-значений гипотез $h_{i,j}:\rho^{i,j}=0$ на группы $S=(0.001; 0.01), I=(0.44;0.71), N=(0,9;0.95)$. Такое разбиение означает, что гипотезы $h_{i,j}:\rho^{i,j}=0$ с p-значениями $0.001; 0.01$, и альтернативы $k_{i,j}:\rho^{i,j}\neq 0$ с р-значениями $0,1;0.05$ отвергаются.  Отвержение гипотезы $k_{i,j}:\rho^{i,j}\neq 0$ с p-значением $0.1$ процедурой $mB(\alpha)$ означает, что $0.1=\frac{\alpha}{6}$, поэтому $\alpha=0,6$, т.е. $P^{\star}=1-\alpha=0.4$ и $P(L(X)\subset J_t(\theta)\subset U(X))\geq 0.4$, где $L(x)=\{(2,3),(2,4)\},U(x)=\{(2,3),(2,4),(1,3),(1,4)\}, \overline{U}(x)=\{(1,2),(3,4)\}$. Размер доверительного множества в этом случае равен $2^2=4$.

Let's consider the partition of p-values for hypotheses \(h_{i,j}:\rho^{i,j}=0\) into groups:  
- \(S = (0.001; 0.01)\),  
- \(I = (0.44; 0.71)\),  
- \(N = (0.9; 0.95)\).  

This partition implies that hypotheses \(h_{i,j}:\rho^{i,j}=0\) with p-values 0.001 and 0.01, as well as the alternatives \(k_{i,j}:\rho^{i,j}\neq 0\) with p-values 0.1 and 0.05, are rejected.  

Rejecting the hypothesis \(k_{i,j}:\rho^{i,j}\neq 0\) with p-value 0.1 using the \(mB(\alpha)\) procedure means that  
\[
0.1 = \frac{\alpha}{6} \rightarrow \alpha = 0.6,
\]
so the confidence level is  
\[
P^{\star} = 1 - \alpha = 0.4.
\]
This indicates that  
\[
P(L(X) \subset J_t(\theta) \subset U(X)) \geq 0.4,
\]
where \(L(x) = \{(2,3), (2,4)\}\),  
\(U(x) = \{(2,3), (2,4), (1,3), (1,4)\}\),  
and \(\overline{U}(x) = \{(1,2), (3,4)\}\).  

The size of the confidence set in this case is \(2^2=4\).

Let's consider the partition of p-values for hypotheses \(h_{i,j}:\rho^{i,j}=0\) into groups:  
- \(S = (0; 0.01)\),  
- \(I = (0.44; 0.71; 0.9)\),  
- \(N = (0.95)\).  

This partition implies that hypotheses \(h_{i,j}:\rho^{i,j}=0\) with p-values 0.001 and 0.01, as well as the alternative hypothesis \(k_{i,j}:\rho^{i,j}\neq 0\) with p-value 0.05, are rejected.  

Rejecting the hypothesis \(k_{i,j}:\rho^{i,j}\neq 0\) with p-value 0.05 using the \(mB(\alpha)\) procedure means that  
\[
0.05 = \frac{\alpha}{6} \rightarrow \alpha = 0.3,
\]
so the confidence level is  
\[
P^{\star} = 1 - \alpha = 0.7.
\]
This indicates that  
\[
P(L(X) \subset J_t(\theta) \subset U(X)) \geq 0.7,
\]
where:  
- \(L(x) = \{(2,4)\}\),  
- \(U(x) = \{(2,3), (2,4), (1,3), (1,4)\}\),  
- \(\overline{U}(x) = \{(1,2), (3,4)\}\).  

The size of the confidence set in this case is  
\[
2^3 = 8.
\]
 
\end{example}

\begin{example}
Let us consider example mathematical marks from \cite{Drton2004} where $N=5$, i.e. $M=10$.
Corresponding p-values $\hat{p}^{\varphi}_{i,j}$ of tests $\varphi_{i,j}$ for testing hypotheses $h_{i,j}:\rho^{i,j}=0$ are equal to:
$\hat{p}^{\varphi}_{1,2}=0.02;\hat{p}^{\varphi}_{1,3}=0.29;\hat{p}^{\varphi}_{1,4}=1.00;\hat{p}^{\varphi}_{1,5}=1.00;$
$\hat{p}^{\varphi}_{2,3}=0.095;\hat{p}^{\varphi}_{2,4}=1.00;\hat{p}^{\varphi}_{2,5}=1.00;\hat{p}^{\varphi}_{3,4}=0.00;$
$\hat{p}^{\varphi}_{3,5}=0.01;\hat{p}^{\varphi}_{4,5}=0.18;.$
For $\alpha=0.1$ $\frac{\alpha}{M}=0.01$, $1-\frac{\alpha}{M}=0.99$. Then 
$$U(x)=\{(1,2);(1,3);(1,4);(1,5);(2,3);(2,4);(2,5);(4,5)\}$$
$$\overline{U}(x)=\{(3,4);(3,5)\}$$
$$L(x)=\{(1,4);(1,5);(2,4);(2,5)\}$$
The conclusion regarding the corresponding graphical model is presented in Fig. \ref{Mathematics_marks}.\begin{figure}[h!]
\centering
\includegraphics*[width=0.5\textwidth]{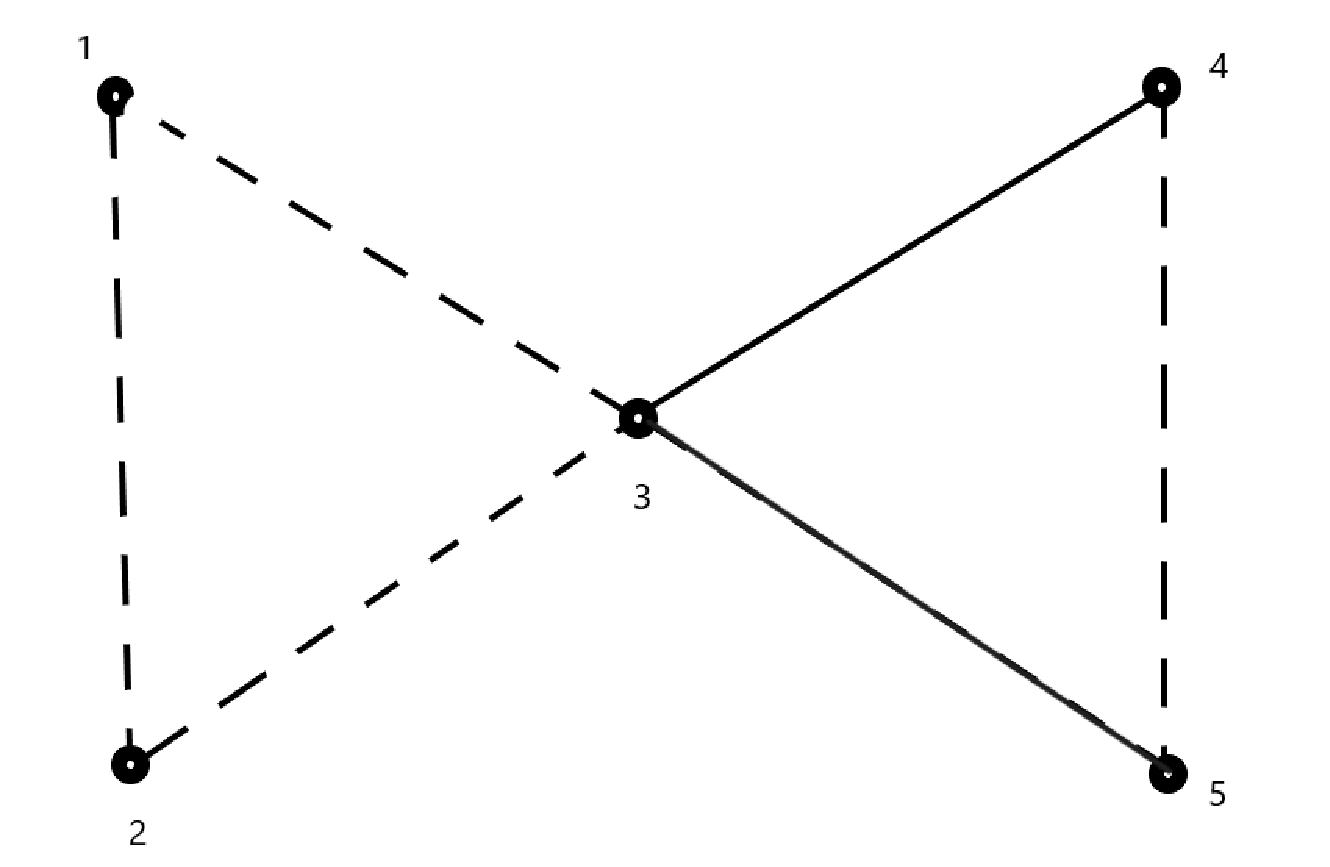}
\caption{Mathematics marks from \cite{Drton2004}, $N=5$. Solid lines correspond to significant edges, while dashed lines correspond to insignificant edges. $P^{\star}=0.9$ 
}\label{Mathematics_marks}
\end{figure}
The set of graphical models that can be obtained by replacing the dashed lines in the graph of Fig. \ref{Mathematics_marks} with solid lines or by removing lines represents a confidence set for the true graphical model at the level \(P^{\star} = 1 - \alpha = 0.9\). The size of this confidence set in this example is \(2^5 = 32\).

In \cite{Drton2004}, a partition of the p-values for hypotheses \(h_{i,j}:\rho^{i,j}=0\) was proposed into groups:  
- \(S = (0; 0.01; 0.02)\),  
- \(I = (0.09; 0.18; 0.29)\),  
- \(N = (1.00; 1.00; 1.00; 1.00)\).  

This partition indicates that hypotheses \(h_{i,j}:\rho^{i,j}=0\) with p-values 0, 0.01, and 0.02, as well as the alternatives \(k_{i,j}:\rho^{i,j}\neq 0\) with p-values 0.0, 0.0, 0.0, and 0.0, are rejected.  

Rejecting the hypothesis \(h_{i,j}:\rho^{i,j}=0\) with p-value 0.02 using the \(mB(\alpha)\) procedure means that  
\[
0.02 = \frac{\alpha}{10} \rightarrow \alpha=0.2,
\]
so the significance level is \(\alpha=0.2\). 
$P^{\star}=1-\alpha=0.8$ and $P(L(X)\subset J_t(\theta)\subset U(X))\geq 0.8$, where 
$$L(x)=\{(1,4);(1,5);(2,4);(2,5)\},U(x)=\{(1,4);(1,5);(2,4);(2,5);(1,3);(2,3);(4,5)\}$$
$$\overline{U}(x)=\{(3,4);(1,2);(3,5)\}$$ 
Thus, the partition of the inferences about the structure of the undirected Gaussian graphical model proposed in \cite{Drton2004} corresponds to a confidence level \(P^{\star} = 0.8\). Note that decreasing \(P^{\star}\) from 0.9 to 0.8 results in the addition of two statistically significant conclusions regarding the validity of the conditional independence hypotheses. At the same time, the set of statistically significant conclusions about the presence of edges remains unchanged.
\end{example}
\begin{example}
Let us consider example Fowl bones from \cite{Drton2004} where $N=6$, i.e. $M=15$.
Corresponding p-values $\hat{p}^{\varphi}_{i,j}$ of tests $\varphi_{i,j}$ for testing hypotheses $h_{i,j}:\rho^{i,j}=0$ are equal to:
$\hat{p}^{\varphi}_{1,2}=0.00;\hat{p}^{\varphi}_{1,3}=0.99;\hat{p}^{\varphi}_{1,4}=0.99;\hat{p}^{\varphi}_{1,5}=1.00;\hat{p}^{\varphi}_{1,6}=0.92$
$\hat{p}^{\varphi}_{2,3}=0.03;\hat{p}^{\varphi}_{2,4}=0.68;\hat{p}^{\varphi}_{2,5}=1.00;\hat{p}^{\varphi}_{2,6}=0.82;\hat{p}^{\varphi}_{5,6}=0.00;$
$\hat{p}^{\varphi}_{3,4}=0.00;\hat{p}^{\varphi}_{3,5}=0.07;\hat{p}^{\varphi}_{3,6}=0.98;\hat{p}^{\varphi}_{4,5}=0.59;\hat{p}^{\varphi}_{4,6}=0.00;$
For $\alpha=0.1$ $\frac{\alpha}{M}=0.0067$, $1-\frac{\alpha}{M}=0.993$. Therefore 
$$U(x)=\{(1,3);(1,4);(1,5);(1,6);(2,3);(2,4);(2,5);(2,6);(3,5);(3,6);(4,5)\}$$
$$\overline{U}(x)=\{(3,4);(1,2);(4,6);(5,6)\}$$ 
$$L(x)=\{(1,5);(2,5)\}$$
The conclusion regarding the corresponding graphical model is presented in Fig.\ref{Fowl_bones}.
\begin{figure}[h!]
\centering
\includegraphics*[width=0.5\textwidth]{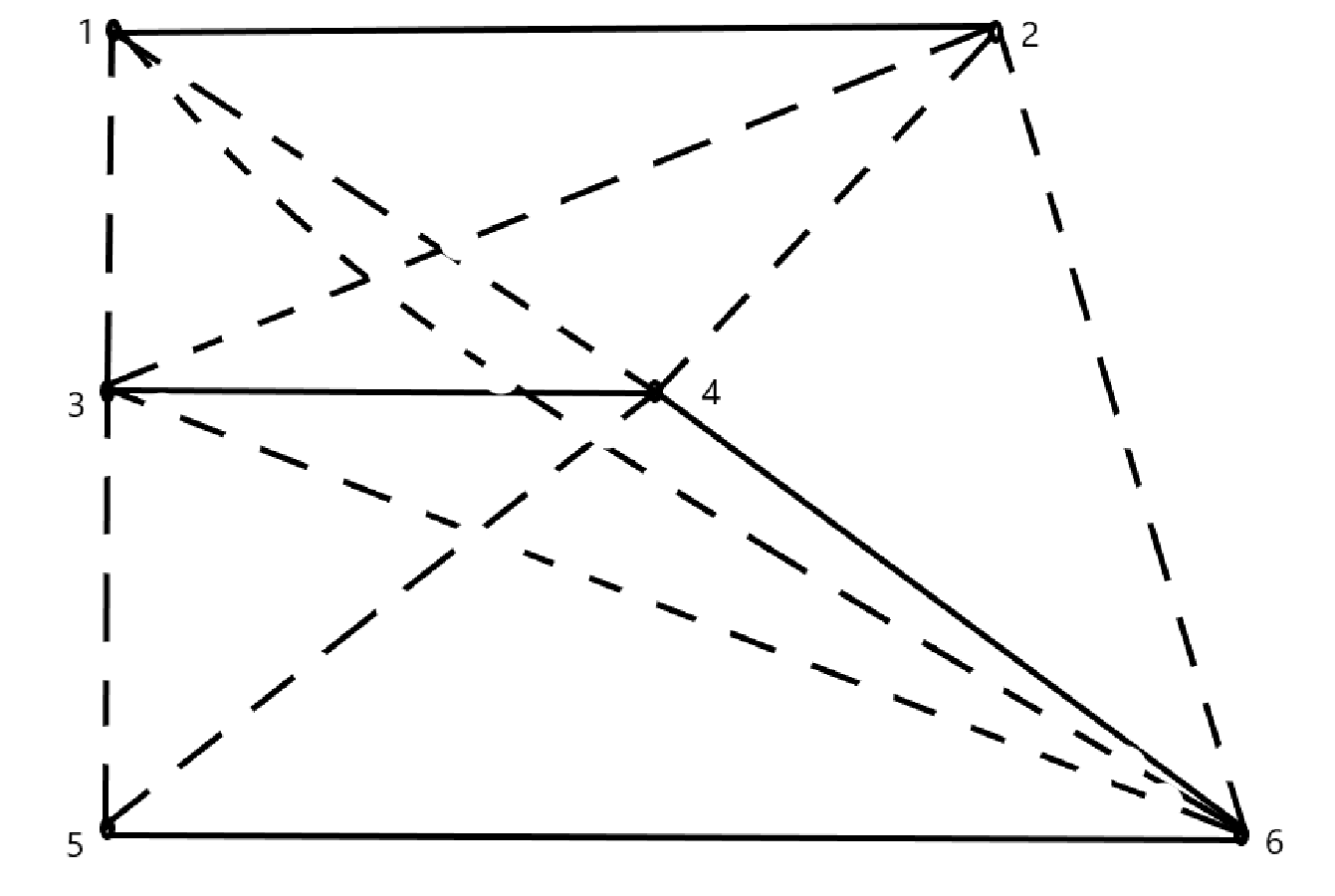}
\caption{Fowl bones from \cite{Drton2004}, $N=6$. Solid lines correspond to significant edges, while dashed lines correspond to insignificant edges. $P^{\star}=0.9$
}\label{Fowl_bones}
\end{figure}
The set of graphical models that can be obtained by replacing the dashed lines in the graph of Fig. \ref{Fowl_bones} with solid lines or by removing lines represents a confidence set for the true graphical model at the level $P^{\star}=1-\alpha=0.9$. The size of this confidence set in this example is $2^9=512$. 

In \cite{Drton2004}, a partition of the p-values for hypotheses \(h_{i,j}:\rho^{i,j}=0\) was proposed into groups:  
\[
S = (0; 0; 0; 0; 0.03), \quad I = (0.07), \quad N = (0.59; 0.68; 0.82; 0.92; 0.98; 0.99; 0.99; 1.00; 1.00).
\]
This partition indicates that hypotheses \(h_{i,j}:\rho^{i,j}=0\) with p-values \(0; 0; 0; 0; 0.03\), as well as the alternatives \(k_{i,j}:\rho^{i,j}\neq 0\) with p-values  
\(0.41,\, 0.32,\, 0.18,\, 0.08,\, 0.02,\, 0.01,\, 0.01,\, 0,\, 0\), are rejected.

Rejecting the alternative hypothesis \(k_{i,j}:\rho^{i,j}\neq 0\) with p-value \(0.41\) using the \(mB(\alpha)\) procedure means that  
\[
0.41 = \frac{\alpha}{15} \rightarrow \alpha=6.15,
\]
which exceeds any conventional significance level (\(\alpha \leq 1\)). Therefore, for the proposed partition, the confidence level is  
\[
P^{\star} = 0.
\]
Let us consider the partition of the p-values into groups 
$$S=(0; 0; 0; 0; 0.03), I=(0.07; 0,59; 0,68; 0,82; 0,92),$$
$$N=(0,98; 0,99; 0,99; 1,00; 1,00)$$ 
This partition indicates that hypotheses \(h_{i,j}:\rho^{i,j}=0\) with p-values \(0; 0; 0; 0; 0.03\), as well as the alternatives \(k_{i,j}:\rho^{i,j}\neq 0\) with p-values \(0.02; 0.01; 0.01; 0; 0\), are rejected.

Rejecting the hypothesis \(h_{i,j}:\rho^{i,j}=0\) with p-value \(0.03\) using the \(mB(\alpha)\) procedure implies that  
\[
0.03 = \frac{\alpha}{15} \rightarrow \alpha=0.45,
\]
and therefore the confidence level is  
\[
P^{\star} = 1 - \alpha = 0.55,
\]
which means that  
\[
P(L(X) \subset J_t(\theta) \subset U(X)) \geq 0.55.
\]
where \(L(X)\) and \(U(X)\) are 
$$L(x)=\{(1,3);(1,4);(1,5);(2,5);(3,6)\}$$
$$U(x)=\{(1,3);(1,4);(1,5);(1,6);(2,4);(2,5);(2,6);(3,5);(3,6); (4,5)\}$$
$$\overline{U}(x)=\{(1,2);(2,3);(3,4);(4,6);(5,6)\}$$
Let us consider the partition of the p-values into groups 
$$S=(0; 0; 0; 0), I=(0.03; 0.07; 0,59; 0,68; 0,82; 0,92),$$
$$N=(0,98; 0,99; 0,99; 1,00; 1,00)$$ 
This partition indicates that the hypotheses \(h_{i,j}:\rho^{i,j}=0\) with p-values \(0; 0; 0; 0\), as well as the alternatives \(k_{i,j}:\rho^{i,j}\neq 0\) with p-values \(0.02; 0.01; 0.01; 0; 0\), are rejected.
Rejecting the alternative hypothesis \(k_{i,j}:\rho^{i,j}\neq 0\) with p-value \(0.02\) using the \(mB(\alpha)\) procedure means that  
\[
0.02 = \frac{\alpha}{15} \rightarrow \alpha=0.3,
\]
and therefore, the confidence level is  
\[
P^{\star} = 1 - \alpha = 0.7,
\]
and %the probability that the true graphical model lies within the interval \([L(X), U(X)]\) is at least 0.7:  
\[
P(L(X) \subset J_t(\theta) \subset U(X)) \geq 0.7.
\]$$L(x)=\{(1,3);(1,4);(1,5);(2,5);(3,6)\}$$
$$U(x)=\{(1,3);(1,4);(1,5);(1,6);(2,3);(2,4);(2,5);(2,6);(3,5);(3,6); (4,5)\}$$
$$\overline{U}(x)=\{(1,2);(3,4);(4,6);(5,6)\}$$
\end{example}

\section{Conclusion}\label{conclusion}

In the present work, the approach to dividing inferences into statistically significant and non-significant, proposed in \cite{Koldanov_2023}, has been extended to address a broader class of problems related to the selection of graphical models. It is demonstrated how this development leads to the construction of confidence sets for graphical models. 

The equivalence between procedures for constructing confidence sets for graphical models and procedures for simultaneous testing of hypotheses and alternatives about the structure of these models has been established. It is shown that, under the condition of free combination of hypotheses and alternatives, a simple generalization of the closure method under broad conditions results in single-step procedures. 

The obtained results have been applied to the division of inferences into significant and uncertain ones in the problem of undirected Gaussian graphical model selection. Additionally, a response has been provided to the question of at which significance level one can trust the results obtained using the SINful approach, as presented in \cite{Drton2004} and \cite{Drton2008}.

\section{Appendix}\label{Appendix}
\subsection{Proof of lemma \ref{upper_set_lemma_new}}
\begin{proof}
If $J_t(\theta)=\emptyset$ or $U(x)=J$ then the lemma is obvious. Let $J_t(\theta)\neq\emptyset, U(x)\subset J$. 

Define set $\overline{U}(x)=\{i:\varphi_i(x)=1\},\ \ i=1,\ldots,M$ --- set of indices of rejected hypotheses $h_i$ by procedure $(\varphi_1,\varphi_2,\ldots,\varphi_M)$. It is evident 
$$U(x)\cup\overline{U}(x)=J, U(x)\cap\overline{U}(x)=\emptyset.$$
Since $J_t(\theta)\subseteq J$, then $J_t(\theta)\subseteq U(x)\cup\overline{U}(x)$.
Therefore
$$1=P_{\theta}\left(J_t(\theta)\subseteq U(X)\cup\overline{U}(X)\right)=P_{\theta}(J_t(\theta)\subseteq U(X))+$$
$$+P_{\theta}\left(\bigcup_{\pi}\left(A_{\pi}\subseteq U(X)\right)\cap\left(B_{\pi}\subseteq\overline{U}(X)\right)\right)+P_{\theta}\left(J_t(\theta)\subseteq\overline{U}(X)\right)$$
where 

$\pi$ is a set of partitions $J_t(\theta)$ by two nonempty sets $A_{\pi}$ and $B_{\pi}$,

$A_{\pi}$ is a set of indices of true hypotheses accepted by procedure $(\varphi_1,\varphi_2,\ldots,\varphi_M)$,

$B_{\pi}$ is a set of indices of true hypotheses rejected by procedure $(\varphi_1,\varphi_2,\ldots,\varphi_M)$.

$$P_{\theta}\left(\bigcup_{\pi}\left(A_{\pi}\subseteq U(X)\right)\cap\left(B_{\pi}\subseteq\overline{U}(X)\right)\right)+P_{\theta}\left(J_t(\theta)\subseteq\overline{U}(X)\right)\leq$$
$$\leq P_{\theta}\left(\bigcup_{\pi}\left(B_{\pi}\subseteq\overline{U}(X)\right)\right)+P_{\theta}\left(J_t(\theta)\subseteq\overline{U}(X)\right)$$

Since procedure $(\varphi_1,\varphi_2,\ldots,\varphi_M)$ for testing hypotheses $h_i,\ \ i=1,\ldots,M$ has to control $FWER\leq\alpha$ then 
$$P_{\theta}\left(\bigcup_{\pi}\left(B_{\pi}\subseteq\overline{U}(X)\right)\right)+P_{\theta}\left(J_t(\theta)\subseteq\overline{U}(X)\right)\leq\alpha$$
Therefore
$$1\leq P_{\theta}(J_t(\theta)\subseteq U(X))+\alpha$$
or
$$P_{\theta}(U(X)\supseteq J_t(\theta))\geq 1-\alpha, \ \ \forall \theta\in\Theta.$$
 \end{proof}

\subsection{Proof of the lemma \ref{low_set_lemma_new}}
\begin{proof}
If $J_t(\theta)=\emptyset$ then (\ref{lower_set_lemma_new_equation}) means that 
$$P_{\theta}(L(X)=\emptyset)\geq 1-\alpha.$$
This is correct since procedure $(\psi_1,\psi_2,\ldots,\psi_M)$
control FWER in strong sense i.e. under condition that all alternatives are true.%, то есть и при условии, что все альтернативные гипотезы истинны.
Let $J_t(\theta)\neq\emptyset.$

$$P_{\theta}(L(X)\subseteq J_t(\theta))=P_{\theta}(L(X)\subseteq J_t(\theta)|L(X)=\emptyset)P_{\theta}(L(X)=\emptyset)+$$
$$+P_{\theta}(L(X)\subseteq J_t(\theta)|L(X)\neq\emptyset)P_{\theta}(L(X)\neq\emptyset)=$$
$$=P_{\theta}(L(X)=\emptyset)+P_{\theta}(L(X)\subseteq J_t(\theta)|L(X)\neq\emptyset)P_{\theta}(L(X)\neq\emptyset)$$
since 
$$P_{\theta}(L(X)\subseteq J_t(\theta)|L(X)=\emptyset)=1 \ \ \forall\theta\in\Theta.$$

%если 
%$$P_{\theta}(L(X)\subseteq J_t(\theta)|L(X)\neq\emptyset)\geq 1-\alpha$$
%Покажем, что в условиях леммы справедливо неравенство
Let's show that the inequality below holds under the conditions of the lemma.
\begin{equation}\label{lower_set_lemma_new_equation_sence_of_the_lemma}
P_{\theta}(L(X)\subseteq J_t(\theta)|L(X)\neq\emptyset)\geq 1-\alpha \ \ \forall\theta\in\Theta.
\end{equation}

It is evident from definitions of the sets $L(x), J$ 
$$P_{\theta}(L(X)\subseteq J|L(x)\neq\emptyset)=1$$
Therefore
$$1=P_{\theta}(L(X)\subseteq J|L(X)\neq\emptyset)=P_{\theta}(L(X)\subseteq J_t(\theta)\cup J_f(\theta)|L(X)\neq\emptyset)=$$
$$=P_{\theta}(L(X)\subseteq J_t(\theta)|L(X)\neq\emptyset)+P_{\theta}\left(\bigcup_{\pi}\left[(A_{\pi}\subseteq J_t(\theta))\cap(B_{\pi}\subseteq J_f(\theta))\right]|L(X)\neq\emptyset\right)+$$
$$+P_{\theta}(L(x)\subseteq J_f(\theta)|L(X)\neq\emptyset)$$
where
$\pi=\pi(x)$ is a set of partitions of $L(x)$ by two sets: $A_{\pi}$ - the set of indices of correctly rejected alternative hypotheses \(k_i\) (included in the set of indices of true hypotheses \(h_i\)), and \(B_{\pi}\) the set of indices of incorrectly rejected alternative hypotheses \(k_i\) (included in the set of indices of false hypotheses \(h_i\)).
$$P_{\theta}(L(X)\subseteq J_t(\theta)|L(X)\neq\emptyset)+P_{\theta}\left(\bigcup_{\pi}\left[(A_{\pi}\subseteq J_t(\theta))\cap(B_{\pi}\subseteq J_f(\theta))\right]|L(X)\neq\emptyset\right)+$$
$$+P_{\theta}(L(X)\subseteq J_f(\theta)|L(X)\neq\emptyset)\leq$$
$$\leq P_{\theta}(L(X)\subseteq J_t(\theta)|L(X)\neq\emptyset)+P_{\theta}\left(\bigcup_{\pi}B_{\pi}\subseteq J_f(\theta)|L(X)\neq\emptyset\right)+$$
$$+P_{\theta}(L(X)\subseteq J_f(\theta)|L(X)\neq\emptyset)$$
Since the procedure \(\psi = (\psi_1, \ldots, \psi_M)\) for simultaneous testing of the alternative hypotheses \(k_i\) controls the family-wise error rate \(FWER \leq \alpha\) in the strong sense, then
$$P_{\theta}\left(\bigcup_{\pi}B_{\pi}\subseteq J_f(\theta)|L(X)\neq\emptyset\right)+P_{\theta}(L(X)\subseteq J_f(\theta)|L(X)\neq\emptyset)\leq\alpha.$$
Therefore
$$1\leq P_{\theta}(L(X)\subseteq J_t(\theta)|L(X)\neq\emptyset)+\alpha,$$
or
$$P_{\theta}(L(X)\subseteq J_t(\theta)|L(X)\neq\emptyset)\geq 1-\alpha.$$
%Если 
%$$L(x)=\emptyset$$
%то утверждение очевидно.
Therefore
$$P_{\theta}(L(X)\subseteq J_t(\theta))=P_{\theta}(L(X)=\emptyset)+P_{\theta}(L(X)\subseteq J_t(\theta)|L(X)\neq\emptyset)P_{\theta}(L(X)\neq\emptyset)\geq $$
$$\geq (1-\alpha)P_{\theta}(L(X)=\emptyset)+(1-\alpha)P_{\theta}(L(X)\neq\emptyset)=1-\alpha$$
 
%Попробовать доказать эквивалентность, а именно: если какоето множество является нижней границей, то процедура его построения может рассматриваться как процедура одновременной проверки многих гипотез с контролем $FWER\leq\alpha$. Разобраться с условиями свободной комбинации. Отдельно рассмотреть случай $N=1$.
\end{proof}
\subsection{Proof of the theorem \ref{theorem_for_upper_low_bounds_construction}}
\begin{proof}
If $J_t(\theta)=J$ then (\ref{simult_upper_low_bounds_new_theorem}) follows from lemma \ref{upper_set_lemma_new}.

If $J_t(\theta)=\emptyset$ then (\ref{simult_upper_low_bounds_new_theorem}) follows from lemma \ref{low_set_lemma_new}.

Suppose $J_t(\theta)\neq\emptyset, J_t(\theta)\subset J$.

Let's show that the condition (\ref{comp_condition_new}) guarantees that for all \( x \), the set \( L(x) \) is a subset of \( U(x) \).
 
Since
$$L(x)=\{i:\psi_i=1\}=\{i:\varphi_i=0,\psi_i=1\}\cup\{i:\varphi_i=1,\psi_i=1\}$$
$$U(x)=\{i:\varphi_i=0\}=\{i:\varphi_i=0,\psi_i=0\}\cup\{i:\varphi_i=0,\psi_i=1\}$$
then from (\ref{comp_condition_new}) one obtain
\begin{equation}\label{low_bound_general_form}
L(x)=\{i:\varphi_i=0,\psi_i=1\}
\end{equation}
\begin{equation}\label{upper_bound_general_form}
U(x)=\{i:\varphi_i=0,\psi_i=0\}\cup\{i:\varphi_i=0,\psi_i=1\}
\end{equation}
and
$$\forall x\ \ L(x)\subseteq U(x).$$
Therefore, the expression \( L(x) \subseteq J_t(\theta) \subseteq U(x) \) is equivalent to the simultaneous fulfillment of two conditions
$L(x)\subseteq J_t(\theta)$, $J_t(\theta)\subseteq U(x)$. 

Since the procedure \(\delta = (\varphi_1, \psi_1, \varphi_2, \psi_2, \ldots, \varphi_M, \psi_M)\) for simultaneous testing of hypotheses \(h_1, k_1, h_2, k_2, \ldots, h_M, k_M\) controls the family-wise error rate \(FWER \leq \alpha\) in the strong sense, then
\begin{equation}\label{meaning_delta_control_FWER}
P_{\theta}\left(\left(\bigcup_{i\in J_t(\theta)}\{\varphi_i(X)=1\}\right)\cup\left(\bigcup_{i\in J_f(\theta)}\{\psi_i(X)=1\}\right)\right)\leq\alpha, \ \ \forall \theta\in\Theta
\end{equation}
The event \(\left(\bigcup_{i \in J_t(\theta)} \{\varphi_i(X) = 1\}\right)\) means that at least one true hypothesis is rejected, i.e., that
$\exists k:k\in J_t(\theta)\cap \overline{U}$.

The event $\left(\bigcup_{i\in J_f(\theta)}\{\psi_i(X)=1\}\right)$ means that at least one true alternative is rejected, i.e. 
$\exists j:j\in L\cap J_f(\theta)$.

Therefore (\ref{meaning_delta_control_FWER}) is equivalent to 
$$P_{\theta}(\{\exists j\in L\cap J_f(\theta)\}\cup\{\exists k\in J_t(\theta)\cap\overline{U}\})\leq\alpha$$
Since $\alpha=1-P^{\star}$ then
$$P_{\theta}(L\subseteq J_t(\theta)\subseteq U)=1-P_{\theta}(\{L\not\subseteq J_t(\theta)\}\cup\{J_t(\theta)\not\subseteq U\})=$$
$$=1-P_{\theta}(\{\exists j\in L\cap J_f(\theta)\}\cup\{\exists k\in J_t(\theta)\cap\overline{U}\})\geq P^{\star}$$
\end{proof}                         
\subsection{Proof of the theorem \ref{equivalence_mult_proc_suff_set}}
\begin{proof}
Let us denote \(\delta'\) as the procedure for constructing \(L(x)\) and \(U(x)\), such that%$$I_i=\left\{\
%					\begin{array}{ll}
%					 1,& i\in L\\
%					 0,& i\in \overline{L}\cap U\\
%					-1,& i\in \overline{U}
%					\end{array}
%		  \right.
%$$
$L(x)\subseteq U(x)$ and the condition below is satisfied
$$P_{\theta}(L\subset J_t(\theta)\subset U)\geq P^{\star}, \ \ \forall \theta\in\Theta$$
Since $L(x)\subseteq U(x)$, then the latter is equivalent to
$$
P_{\theta}(L(x)\subseteq J_t(\theta)\mbox{ and } J_t(\theta)\subseteq U(x))\geq P^{\star}
$$
or
$$
P_{\theta}(L(x)\not\subseteq J_t(\theta)\mbox{ or } J_t(\theta)\not\subseteq U(x))\leq 1-P^{\star}=\alpha
$$
or
$$P_{\theta}(\{\exists i\in L\cap J_f(\theta)\}\cup\{\exists j\in J_t(\theta)\cap\overline{U}\})\leq\alpha$$
The latter implies that the probability that at least one index from \(L(x)\) is an index of a false hypothesis, or at least one index from \(\overline{U}(x)\) is an index of a true hypothesis, is bounded by the value \(\alpha\). Consequently, the procedure \(\delta'\) is a simultaneous testing procedure for hypotheses \(h_1, \ldots, h_M\) and alternatives \(k_1, \ldots, k_M\), which controls the family-wise error rate (FWER) at level \(\alpha = 1 - P^{\star}\) in the strong sense.%Последнее равносильно 
%$$P_{\theta}(\{\exists j\in L\cap J_f(\theta)\}\cup\{\exists k\in J_t(\theta)\cap\overline{U}\})\leq\alpha$$
%или
%$$
%P_{\theta}\left(\left(\bigcup_{i\in J_t(\theta)}\{\varphi_i(X)=1\}\right)\cup\left(\bigcup_{i\in J_f(\theta)}\{\psi_i(X)=1\}\right)\right)\leq\alpha, \ \ \forall \theta\in\Theta
%$$
\end{proof}
\subsection{Generalized Closure Method}
\begin{teo}\label{generalization_closure_principle_FWER_control_theorem}
If, in the closure of the original family of hypotheses \(\{h_1, \ldots, h_K\}\), some intersection hypotheses \(H_Q = \bigcap_{i \in Q} h_i\) are empty (\(H_Q = \bigcap_{i \in Q} h_i = \emptyset\)), then the procedure for simultaneous testing of hypotheses \(h_1, \ldots, h_K\), constructed via the generalized closure method, controls the probability of at least one erroneous rejection at level \(\alpha\), i.e., \(FWER \leq \alpha\), for any number of true hypotheses.
\end{teo}
%\begin{note}
%Подчеркнем, что отвержение любой пустой гипотезы пересечения означает, что при последовательной реализации обобщенного метода замыкания процедура не остановится на шагах, предшествующих проверке $H_Q$.
%\end{note}
\begin{proof}
Let hypotheses \(h_i, i \in Q, Q \subset \{1, 2, \ldots, K\}\), be true, and let \(H_Q = \bigcap_{i \in Q} h_i\) be the corresponding intersection hypothesis. Then, the hypothesis \(H_Q\) is obviously non-empty and true.% В соответствии с принципом замыкания хотя бы одна гипотеза $h_i:i\in P$ отвергается (т.е. $\varphi_i=1$) тогда и только тогда, когда все гипотезы $H_Q:i\in Q, Q\subset\{1,2,\ldots,N\}$ отвергаются. 
The event \(A = \{\mbox{rejecting at least one true hypothesis } h_i\}\) can be written as:

\[
A = \{\varphi_{1,2,\ldots,K} = 1\} \cap \ldots \cap \{\varphi_{i_1,i_2,\ldots,i_m} = 1\} \cap \ldots \cap \{\varphi_Q=1\} \cap \ldots \cap \left( \bigcup_{i \in Q} \{\varphi_i=1\} \right),
\]

where \( \{i_1, i_2, \ldots, i_m\} \subseteq \{1, 2, \ldots, K\} \).
If for some set \( \{i_1, i_2, \ldots, i_m\} \), the corresponding intersection hypothesis

\[
H_{i_1,i_2,\ldots,i_m} = \bigcap_{i \in \{i_1,i_2,\ldots,i_m\}} h_i = \emptyset,
\]

then, according to the generalized closure method,

\[
\varphi_{i_1,i_2,\ldots,i_m} = 1.
\]

The event \(A = \{\mbox{rejecting at least one true hypothesis } h_i\}\) is contained within the event "rejecting \(H_Q\) ", i.e.,

\[
A \subseteq \{\varphi_Q=1\}.
\]

Therefore,

\[
P(A) \leq P(\varphi_Q=1).
\]

Since all hypotheses are tested at level \(\alpha\), meaning

\[
P_{H_Q}(\varphi_Q=1) = \alpha,
\]

it follows that

\[
P(A) \leq P(\varphi_Q=1) = \alpha,
\]

and thus the family-wise error rate (FWER) satisfies

\[
FWER = P(\mbox{rejecting at least one true hypothesis}) \leq \alpha.
\]
\end{proof}

\subsection{Proof of the theorem \ref{closure_equivalent_bonferroni_theorem}}\label{proof_closure_equivalent_bonferroni_theorem}
\begin{proof}
Without loss of generality consider $h_1$. %Не уменьшая общности , предположим, что гипотеза $h_1$ истинна. 
The hypothesis \(h_1\) is rejected by the generalized closure method if and only if all non-empty hypotheses \(H_{J_1,J_2}\), such that \(\{1\} \subseteq J_1\), are rejected. Let \(J_1 = \{1, i_2, \ldots, i_{|J_1|}\}\), \(J_2 = \{j_1, \ldots, j_{|J_2|}\}\), with \(|J_1 \cup J_2| = L\) and \(J_1 \cap J_2 = \emptyset\). Clearly, \(1 \leq L \leq M\).%Тогда для того, чтобы гипотеза $h_1$ была отвергнута, необходимо, чтобы были одновременно отвергнуты все гипотезы $H_{J_1,J_2}$, такие что $\{1\}\subset J_1$ и $|J_1\cup J_2|=L$. 
%Легко видеть, что при фиксированном $L$ число гипотез $H_{J_1,J_2}$, таких что $\{1\}\subset J_1$, равно $2^{L-1}$.   

In accordance with condition 2 of the theorem, the tests for evaluating \(H_{J_1,J_2}\) are union-intersection tests, which have the form \cite{Hochberg_1987}:
$$
\delta_{J_1,J_2}=\left\{\
									\begin{array}{ll}
									1,& \min_{i\in J_1,j\in J_2}(\hat{p}_{h_i},\hat{p}_{k_j})<\alpha(L)\\
									0,&\mbox{ otherwise }
									\end{array}
								 \right.
$$
where $\alpha(L)$ is defined from
\begin{equation}\label{union_intersection_test_hypotheses_alternative_L}
\begin{array}{l}
P_{\theta}(\min_{i\in J_1,j\in J_2}(\hat{p}_{h_i},\hat{p}_{k_j})<\alpha(L))=\alpha\\
\end{array}
\end{equation}
$$\theta=(\theta_1,\theta_{i_2},\ldots,\theta_{i_{|J_1|}},\theta_{j_1},\ldots,\theta_{j_{|J_2|}}),\theta_i\in [ \Theta_i ] \cap [\Theta_i^c ] , i\in J_1\cup J_2,$$ 
$[ \Theta_i ] \mbox{ - closure of the set }\Theta_i.$
Since $\alpha(L)$ does not depend from $J_1$ then (\ref{union_intersection_test_hypotheses_alternative_L}) is equivalent to
\begin{equation}\label{union_intersection_test_hypotheses_alternative_L_equivalent}
P_{\theta}(\min_{i\in J_1}(\hat{p}_{h_i})<\alpha(L))=\alpha
\end{equation}
where $|J_1|=L, J_2=\emptyset, \theta=(\theta_1,\theta_{i_2},\ldots,\theta_{i_{L}})$.
%при любом выборе $J_1, J_2:|J_1\cup J_2|=L$ распределение 
%При этом предполагается, что распределение $\hat{p}_{h_i},\hat{p}_{k_i}$ зависит только от $\theta_i$. 

%то гипотеза $H_{J_1,J_2}$ отвергается, если хотя бы одно из p-значений $\hat{p}_{h_i}$, соответствующих гипотезам, индексы которых входят в $J_1$, или хотя бы одно из p-значений $\hat{p}_{k_i}$, соответствующих альтернативам, индексы которых входят в $J_2$, не превышают $\alpha'$, где $\alpha$ --- заданное ограничение на FWER.   

For \(L > 1\) and any index \(m \in \{2, \ldots, M\}\), the set of hypotheses \(H_{J_1,J_2} : \{1\} \subset J_1\) includes both intersection hypotheses containing \(h_m\) and intersection hypotheses containing \(k_m\), where \(m \neq 1\), and \(m \in J_1 \cup J_2\), with \(J_1 \cap J_2 = \emptyset\). 

According to condition 3 of the theorem, the p-values \(\hat{p}_{h_m}\) and \(\hat{p}_{k_m}\) cannot both be less than \(\alpha(L)\) simultaneously. 

Therefore, all hypotheses \(H_{J_1,J_2}\) such that \(\{1\} \subset J_1\) and \(|J_1 \cup J_2| = L\) can be rejected simultaneously if and only if

\[
\hat{p}_{h_1} < \alpha(L).
\]

From equation (\ref{union_intersection_test_hypotheses_alternative_L_equivalent}), it is evident that \(\alpha(L)\) decreases as \(L\) increases. Since \(L \leq M\), the minimum value of \(\alpha(L)\) over all \(L\) is

\[
\min_L \alpha(L) = \alpha(M).
\]
%Все множество гипотез пересечения $H_{J_1,J_2}$, таких что $\{1\}\subset J_1$ можно представить как объединение гипотез $H_{J_1,J_2}$ с $|J_1\cup J_2|=L$, $1\leq L\leq M$, т.е. 
%$$\bigcup_{L=1}^M\bigcup_{\begin{array}{l}\forall J_1,J_2\subset\{1,\ldots,M\}\\
%																	\{1\}\in J_1,|J_1\cup J_2|=L
%													\end{array}}H_{J_1,J_2}$$
%Следовательно, все гипотезы $H_{J_1,J_2}$ будут отвергаться тогда и только тогда, когда $\hat{p}_{h_1}<\min_L\alpha(L)$.  
Therefore, the test for \(h_1\), derived from the generalized closure method, is given by:
$$\varphi_1(x)=\left\{\
									\begin{array}{ll}
									1,& \hat{p}_{h_1}<\alpha(M)\\
									0,& \hat{p}_{h_1}\geq\alpha(M)
									\end{array}
							 \right.,
$$
% что совпадает с модифицированной процедурой Бонферрони.    
Therefore, the procedure for simultaneous testing of hypotheses \(h_1, \ldots, h_M\) and alternatives \(k_1, \ldots, k_M\) is a single-step procedure and has the form:$$\left(\varphi_1,\ldots,\varphi_M,\psi_1(x),\ldots,\psi_M(x)\right)$$
where 
$$\varphi_i(x)=\left\{\
									\begin{array}{ll}
									1,& \hat{p}_{h_i}<\alpha(M)\\
									0,& \hat{p}_{h_i}\geq\alpha(M)
									\end{array}
							 \right.,
$$
$$\psi_i(x)=\left\{\
									\begin{array}{ll}
									1,& \hat{p}_{k_i}<\alpha(M)\\
									0,& \hat{p}_{k_i}\geq\alpha(M)
									\end{array}
							 \right.,
$$
\begin{equation}\label{union_intersection_test_hypotheses_alternative_M_equivalent}
P_{\theta}(\min_{i\in\{1,\ldots,M\}}(\hat{p}_{h_i})<\alpha(M))=\alpha
\end{equation}

\end{proof}

\end{document}